\documentclass[runningheads]{llncs}

\usepackage{capt-of}
\usepackage{amsmath}
\usepackage{comment}
\usepackage{amsfonts}
\usepackage{bbm}
\usepackage[linesnumbered,ruled,vlined,noend]{algorithm2e}
\NoCaptionOfAlgo % hides "Algorithm 1" in our pseudocode

\usepackage{mathdots}
\usepackage{thmtools} 
\usepackage{thm-restate}
\usepackage{hyperref}
\usepackage{times}
\usepackage[normalem]{ulem}
\usepackage{url}
\usepackage{listings}
\usepackage{graphicx}
\usepackage{xspace}
\usepackage[shortlabels]{enumitem} % Max: pulling bullets left
\usepackage[skins]{tcolorbox}
\usepackage{subcaption}
\usepackage{appendix}
\usepackage{wrapfig}%
\usepackage{authblk} % need for authorship list format

\usepackage{tikz}
\newcounter{factcounter}
\newenvironment{fact}[1][]
  {\refstepcounter{factcounter}\par\noindent\textbf{Fact \thefactcounter: }\itshape #1}

\newcommand{\Sample}{\textsc{Collect-Sample}\xspace}
\newcommand{\Diagnosis}{\textsc{Diagnosis}\xspace}
\newcommand{\Rundown}{\textsc{RunDown}\xspace}

\newcommand{\poly}{\texttt{poly}\xspace}
\newcommand{\uar}{u.a.r.\xspace}
\newcommand{\CollCost}{\mathcal{C}\xspace}

\newcommand{\currWin}{w_{\mbox{\tiny cur}}\xspace}
\newcommand{\twoActiveSlots}{\mathcal{S}\xspace}

\newcommand{\equalCost}{P\xspace}
\newcommand{\pmax}{p_{\mbox{\tiny max}}\xspace}
\newcommand{\psec}{p_{\mbox{\tiny sec}}\xspace}
\newcommand{\deltaMin}{\Delta_{\min}\xspace}
\newcommand{\con}{\texttt{Con}(t)\xspace}

\newcommand{\defn}[1]{\textbf{\emph{#1}}}
\newcommand{\whp}{w.h.p.\xspace}

\newcommand{\goodrange}{\textsc{Good}\xspace}
\newcommand{\toohigh}{\textsc{High}\xspace}
\newcommand{\toolow}{\textsc{Low}\xspace}
\newcommand{\tootoolow}{\textsc{Rock-Bottom}\xspace}
\newcommand{\uncertainA}{\textsc{Uncertain-Low}\xspace}
\newcommand{\uncertainB}{\textsc{Uncertain-High}\xspace}

\newcommand{\probP}{\text{I\kern-0.15em P}}
\newcommand{\beb}{BEB\xspace}

\newcommand{\STB}{\textsc{Sawtooth-Backoff}\xspace}

\newcommand{\stb}{STB\xspace}

\newcommand{\myVar}{x\xspace}

\newcommand{\LC}{LC\xspace}

\newcommand{\probParam}{n\xspace}
\newcommand{\sampleParam}{w_{\mbox{\tiny{cur}}}\xspace}
\newcommand{\mainAlgorithm}{\textsc{Collision-Aversion Backoff}\xspace}
\newcommand{\malg}{\textsc{CAB}\xspace}

\AtBeginEnvironment{algorithm}{\textnormal}

\newif\ifcomments
\commentstrue
\commentsfalse
\providecommand{\keywords}[1]
{
  \small	
  \textbf{\textit{Keywords---}} #1
}

\begin{document}

\title{Softening the Impact of Collisions in Contention Resolution\vspace{-8pt}}

\author{Umesh Biswas\inst{1} \and
Trisha Chakraborty\inst{2}  \and
Maxwell Young\inst{1}}
\authorrunning{U. Biswas, T. Chakraborty, and M. Young}
% First names are abbreviated in the running head.
% If there are more than two authors, 'et al.' is used.
%
\institute{Department of Computer Science and Engineering, Mississippi State University, MS 39762, USA\\
\email{ucb5@msstate.edu, myoung@cse.msstate.edu}\\
 \and
Amazon Web Services, MN 55401, USA  \\
\email{trichakr@amazon.com}
}

%CCF-2144410

\maketitle              % 

\vspace{-5pt}

\begin{abstract}
Contention resolution addresses the problem of coordinating access to a shared communication channel. Time is discretized into synchronized slots, and a packet  can be sent in any slot. If no packet is sent, then the slot is empty; if a single packet is sent, then it is successful; and when multiple packets are sent at the same time, a collision occurs, resulting in the failure of the corresponding transmissions. In each slot, every packet receives ternary channel feedback indicating whether the current slot is empty, successful, or a collision. 

~~~~~Much of the prior work on contention resolution has focused on optimizing the makespan, which is the number of slots required for all packets to succeed. However, in many modern systems, collisions are also costly in terms of the time they incur, which existing contention-resolution algorithms do not address. 

~~~~~In this paper, we design and analyze a randomized algorithm, \mainAlgorithm (\malg), that optimizes both the makespan and the collision cost. We consider the static case where an unknown $n\geq 2$ packets are initially present in the system, and each collision has a known cost $\CollCost$, where $1 \leq \CollCost \leq n^{\kappa}$ for a known constant $\kappa\geq 0$. With error probability polynomially small in $n$, \malg guarantees that all packets succeed with makespan and a total expected collision cost of $\tilde{O}(n\sqrt{\CollCost})$. We give a lower bound for the class of fair algorithms: where, in each slot, every packet executing the fair algorithm sends with the same probability (and the probability may change from slot to slot). Our lower bound is asymptotically tight up to a $\poly(\log n)$-factor for sufficiently large $\CollCost$.
\end{abstract}
\vspace{-10pt}

\noindent\keywords{Distributed computing, contention resolution, collision cost.}

\section{Introduction}\label{sec:introduction}

Contention resolution addresses the fundamental challenge of coordinating access by devices to a shared resource, which is typically modeled as a multiple access channel. Introduced with the development of ALOHA in the early 1970s \cite{abramson1970aloha}, the problem of contention resolution remains relevant in WiFi networks \cite{bianchi2000performance}, cellular networks \cite{ramaiyan2014information}, and shared-memory systems \cite{ben2017analyzing,dwork1997contention}.

Our problem is described as follows. There are {\boldmath{$n$}} $\geq 2$ devices present at the start of the system, each with a packet to send on the shared channel, and $n$ is {\it a priori} unknown.  For ease of presentation, we adopt a slight abuse of language by referring to packets sending themselves (rather than referring to devices that send packets). Time proceeds in discrete, synchronized \defn{slots}, and each slot has size that can accommodate the transmission of a packet. For any fixed slot, the channel provides ternary feedback, allowing each packet to learn whether $0$ packets were sent, $1$ packet was sent (a \defn{success}), or $2+$ packets were sent (a \defn{collision}). Sending on the channel is performed in a distributed fashion; that is, {\it a priori} there is no central scheduler or coordinator.  Traditionally, the contention-resolution problem focuses on minimizing the number of slots until all $n$ packets succeed, which is referred to as the \defn{makespan}.

However, in many modern systems, collisions also have a significant impact on performance. In the \defn{standard cost model}, each collision can be viewed as a wasted slot that increases the makespan by $1$, since no packet can succeed. Yet, in many settings, a collision wastes {\it more} than a single slot, and we denote this cost by {\boldmath{$\CollCost$}}.\footnote{\scriptsize We may also view $\CollCost$ as an upper bound on the cost of a collision if there is some variance.} This cost can vary widely across different systems. In intermittently-connected mobile wireless networks \cite{zhang2006routing}, failed communication due to a collision may result in the sender having to wait until the intended receiver is again within transmission range. For WiFi networks,  a collision may consume time proportional to the packet size \cite{anderton:windowed}. In shared memory systems,  concurrent access to the same memory location can result in delay proportional to the number of contending processes \cite{ben2017analyzing}. Thus, makespan gives only part of the performance picture, since  collisions may add to the time until all packets succeed. 

In this work, we account for the cost of collisions, in addition to the makespan. As we discuss later (see Section~\ref{sec:technical-overview}), existing contention-resolution algorithms do not perform well in this setting. For instance, randomized binary exponential backoff (\beb) \cite{MetcalfeBo76}, arguably the most popular of all contention-resolution algorithms, incurs a collision cost of $\Omega(n\,\CollCost)$, where $\CollCost$ is the cost of a single collision. New ideas are needed to achieve better performance in this model.

\subsection{Model and Notation}\label{sec:model}

 Our work addresses the case where an $n\geq 2$ number of packets are active at the start of the system. There is no known upper bound on $n$ {\it a priori}, and packets do not have identifiers. A packet is \defn{active} if it is executing a protocol; otherwise, the packet has terminated. Each packet must be successfully sent in a distributed fashion (i.e., there is no scheduler or central authority) on a multiple access channel, which we now describe.  \medskip

%A multiple access channel is a well-known model of a shared communication channel (for examples, see~\cite{chlebus2019energy,willard:loglog,chlebus2005almost,kowalski2005time,de2010distributed}).

\noindent{\bf Multiple Access Channel.} Time is divided into disjoint \defn{slots}, where each slot is sized to accommodate a single packet.  In any slot, a packet may send itself or it may listen. For any fixed slot, if no packet is sent, then the slot is \defn{empty}; if a single packet is sent, then it is \defn{successful}; and if two or more packets are sent, then there it is a \defn{collision} and none of these packets is successful. A packet that transmits in a slot immediately learns of success or failure; if it succeeds, the packet terminates immediately. Any packet that is listening to a slot, learns which of these three cases occurred; this is the standard \defn{ternary feedback model}. Likewise, for any slot, a packet that is sending can also determine the channel feedback; either the packet succeeded (which it learns in the slot) or it failed (from which the packet infers a collision occurred). Packets cannot send messages to each other, and no other feedback is provided by the channel. 
\medskip

\noindent{\bf Metrics.} A well-known measure of performance in prior work is the \defn{makespan}, which is the number of slots required for all $n$ packets to succeed. An equivalent metric in the static case is, \defn{throughput}, which is defined as the time for $n$ successes divided by the makespan. 

In our model, each collision incurs a known cost of {\boldmath{$\CollCost$}}, where $1 \leq \CollCost \leq n^{\kappa}$, for a known constant {\boldmath{$\kappa$}}$ \geq 0$.\footnote{\scriptsize Note that knowing $\CollCost$ and $\kappa$ implies only a lower (not upper) bound on $n$.} Given a contention-resolution algorithm, the pertinent costs are: (i) the \defn{collision cost}, which is the number of collisions multiplied by $\CollCost$, and (ii) the makespan. Often we refer to ``cost'', by which we mean the maximum of (i) and (ii), unless specified otherwise.

%Correspondingly, our notion of throughput is the time for $n$ successful slots divided by (i) plus (ii).

%; the optimal makespan is $n$. %

%We define analogous metrics for makespan and throughput as the total cost and $n$ divided by the total cost.

%We aim to minimize the maximum of the makespan and expected collision cost. %\maxwell{polish: } 

\medskip
\noindent{\bf Notation.} Throughout this manuscript,  {\boldmath{$\lg(\cdot)$}} refers to logarithm base $2$, and {\boldmath{$\ln(\cdot)$}} refers to the natural logarithm. We use {\boldmath{$\log^{y}(x)$}} to denote $(\log(x))^y$, and we use \defn{\poly}{\boldmath{$(x)$}} to denote $x^{y}$ for some constant $y\geq 1$. The notation {\boldmath{$\tilde{O}$}} denotes the omission of a $\poly(\log n)$ factor. We say an event occurs \defn{with high probability (\whp) in} {\boldmath{$n$}} (or simply ``with high probability'') if it occurs with probability at least $1-O(1/n^{c})$, for any tunable constant $c\geq 1$.

For any random variable $X$, we say that \whp  the expected value $E[X] \leq x$ if, when tallying the events in the expected value calculation, the sum total probability of those events where $X>x$ is $O(1/n^{c})$ for any tunable constant $c\geq 1$.  In our analysis, the  expected collision cost holds so long as the subroutine \Diagnosis (described later) gives correct feedback to the packets, and this occurs with high probability.

\subsection{Our Results}\label{sec:related-work}

We design and analyze a new algorithm, \defn{\mainAlgorithm (\malg)}, that solves the static contention resolution problem with high-probability bounds on makespan and the expected collision cost. The following theorems provide our formal result.

\begin{theorem}\label{thm:main-theorem}
    With high probability in $n$, \mainAlgorithm guarantees that all $n$ packets succeed with a makespan of $\tilde{O}(n\sqrt{\CollCost})$, and an expected collision cost of $\tilde{O}(n\sqrt{\CollCost})$. 
\end{theorem}

How does this compare with prior results?  For starters, consider \STB (\stb), which \whp  has an asymptotically-optimal makespan of $\Theta(n)$, but incurs $\Omega(n)$ collisions. The expected cost for  \malg is superior to \stb when $\CollCost$ is at least polylogarithmic in $n$. In  Section \ref{sec:discussion-upper}, we elaborate on how our result fits with previous work on  contention resolution.

In a well-known lower bound in the standard cost model, Willard \cite{willard:loglog} defines and argues about \defn{fair} algorithms. These are algorithms where, in a fixed slot, every active packet sends with the same probability (and the probability may change from slot to slot). Our lower bound applies to fair algorithms as follows.
 
\begin{theorem}\label{thm:main-lower-bound}
Any fair algorithm that \whp has a makespan of $\tilde{O}(n\sqrt{\CollCost})$ has \whp an expected collision cost of $\tilde{\Omega}(n\sqrt{\CollCost})$ for $\CollCost = \Omega(n^2)$.  
\end{theorem}

\noindent Since \malg is fair and guarantees \whp a makespan of $\tilde{O}(n\sqrt{\CollCost})$, our upper bound is asymptotically tight up to a $\poly(\log n)$-factor for sufficiently large $\CollCost$. A more general form of our lower bound is given in Section \ref{thm:general-lower-bound-main}, along with additional discussion.  

\subsection{Why Care About Collision Costs?}\label{sec:discussion-upper}

Here, we discuss the potential value of our result and exploring beyond the  standard cost model for contention resolution. \medskip

\noindent{\bf A Simple Answer.} Prior contention-resolution algorithms that optimize for makespan suffer from many collisions. For example, arguably the most famous backoff algorithm is \beb (despite its sub-optimal makespan \cite{bender:adversarial}), which has $\Omega(n)$ collisions. Similarly, \stb has asymptotically-optimal makespan, but also suffers $\Omega(n)$ collisions \cite{anderton:windowed}. Thus, these algorithms have cost $\tilde{\Theta}(n + n\CollCost)$. In contrast, the expected completion time for \malg is $\tilde{O}(n\sqrt{\CollCost})$, which is superior whenever $\CollCost=\poly(\log n)$.

In the extreme, even a hypothetical algorithm that suffers (say, \whp) a single collision would do poorly by comparison if collisions are sufficiently costly. Specifically, such an algorithm pays $\CollCost$, which is asymptotically worse than the expected completion time of \malg for $\CollCost = \tilde{\omega}(n^2)$.

%Perhaps the above is reasonable motivation (we think it is), but what about other natural models of cost. What if successes also have cost? Below, we argue that the cost of successes is not that interesting.

Perhaps the above is reasonable motivation (we think it is), but why not also consider the cost of successes?\footnote{\scriptsize Empty slots do not seem to warrant such consideration since, by definition, nothing is happening in such slots, and so a per-cost of $1$ makes sense.} Below, we argue that reducing collision cost remains important, and that adding a non-unit cost for successes does not seem interesting from a theory perspective.\medskip

\noindent{\bf Connecting to the Standard Cost Model.} Let us {\it temporarily} consider the implications of a cost model where a success and a collision each have cost $\equalCost \geq 1$. In the standard cost model $\equalCost=1$.  In WiFi networks, both costs are also roughly equal, where $\equalCost\gg 1$  \cite{anderton:windowed}.

%\footnote{\scriptsize Given this unavoidable cost, it does not seem  interesting from a theory perspective to consider a model where a success costs far more than a collision, since the cost from successes would dominate. Furthermore, we are unaware of a practical motivation for such a model.}
% For instance, in wireless networks, both costs can be roughly equal \cite{anderton:windowed}; denote this cost by $\equalCost\geq 1$.  Indeed, the standard cost model for contention resolution assumes equality, with $\equalCost=1$; that is, both successes and collisions each have cost $1$. 

In this context, reconsider \stb's performance. Since, $n$ packets must succeed, a cost of at least $n\equalCost$ is unavoidable, and the additional $\Theta(n\equalCost)$ cost from collisions becomes (asymptotically) unimportant. (Indeed, any algorithm must pay at least  $n\equalCost$ for successes; given this unavoidable cost, it does not seem  interesting from a theory perspective to consider a model where a success costs far more than a collision, since the cost from successes would dominate.)  Likewise,  \malg must also pay at least $n\equalCost$, which is (asymptotically) no better than \stb.

%The point is that considering the cost of only of successes, or where the cost of successes is much larger than that of collisions is not interesting.

Does reducing collision costs matter here? Yes, and the value of our result is best viewed via throughput (recall Section \ref{sec:model}). \stb has $c n$ collisions, for some constant $c>0$, and (ignoring empty slots) its throughput is less than $n\equalCost/(n\equalCost+(cn)\equalCost) < 1/(1+c) < 1 - 1/c$; that is, the throughput is bounded away from $1$ by a constant amount. Doing a similar calculation for \malg, there is $n\equalCost$ cost for $n$ packets to succeed, plus $O(n\sqrt{\equalCost})$, which accounts for collisions and empty slots. Thus, the expected throughput is $E[n\equalCost/T]$ where $T$ is the completion time. Since, $E[1/T] \geq 1/E[T]$, we have  $E[n\equalCost/T] \geq n\equalCost/(n\equalCost + \tilde{O}(n\sqrt{\equalCost}))$ $ >  1-\tilde{O}(1/\sqrt{\equalCost})$. 

Thus, for the standard cost model, with $\equalCost=1$, our result is (unsurprisingly) not an improvement. However, as $\equalCost$ grows, our result provides better throughput in expectation, approaching $1$. There are settings where we expect $\equalCost$ to be large, such as: wireless networks where $\equalCost$ can be commensurate with packet size,  routing in mobile networks where communication failure increases latency \cite{zhang2006routing}, and shared memory systems where concurrent access to the same memory location by multiple processes results in delay \cite{ben2017analyzing}.

\section{Related Work}\label{sec:related}

There is a large body of work on the static case for contention resolution, where $n$ packets arrive together and initiate the
contention resolution algorithm; this is often referred to as the \defn{static} or \defn{batched-arrival} setting.  Bender et al.~\cite{bender:adversarial} analyze the makespan for \beb, as well as other backoff algorithms; surprisingly, they show that \beb has sub-optimal makespan. 

%Probably the most popular algorithm for this setting is \beb. Informally, under \beb each packet attempts to send itself once in a contiguous sequence of slots (a \defn{window}); if the packet succeeds, it departs, otherwise, it tries in the next window, which has double the size of the current window.  
%with the property that successive windows have size that increases monotonically.

%These variants operate similarly to \beb: when the current window of size $w$ is finished, the next window has size $(1+ f(w))w$, where $f(w)=1/\log(w), 1/\log\log(w),$ $1$, and $0$, for \lb, \llb, \beb, and \fb, respectively. Interestingly, the authors show that \beb has suboptimal makespan; in particular, with high probability (\whp) in $n$,  $\Theta(n\log n)$ slots are required for all $n$ packets to succeed. Subsequently, out of all monotonic backoff algorithms, \llb, is optimal, with a makespan of $\Theta(n\log\log n/\log\log\log n)$. 

An optimal backoff algorithm is \stb ~\cite{gereb1992efficient,Greenberg1985randomized}. Since we borrow from \stb to create one of our subroutines (discussed further in Section \ref{sec:our-approach}), so we describe it here. Informally, \stb  works by executing over a doubly-nested loop, where the outer loop sets the current window size $w$ to be double the one used in the preceding outer loop. Additionally, for each such window, the inner loop executes over a \defn{run} of $\lg w$ windows of decreasing size: $w, w/2, w/4, ..., 1$. For each  window, every packet chooses a slot uniformly at random (\uar) to send in.

The static setting has drawn attention from other angles.  Bender et al.~\cite{bender:heterogeneous} examines a model where packets have different sizes under the binary feedback model.  Anderton et al.~\cite{anderton:windowed} provide experimental results and argue that packet size should be incorporated into the definition of makespan, since collisions tend to cost time proportional to packet size.

To the best of our knowledge, our work is the first to propose an algorithm for minimizing collision cost, and so it makes sense to start with the static setting. However, we note the dynamic setting has been addressed (under the standard cost model), where packet arrival times are governed by a stochastic process (see the survey by Chlebus \cite{chlebus2001randomized}). Another direction that the research community has pursued is the application of adversarial queueing theory  (AQT) \cite{borodin2001adversarial} to contention resolution; for examples, see \cite{chlebus2009maximum,chlebus2007stability,anantharamu2010deterministic,aldawsari:adversarial}. Under the AQT setting, packet arrival times are dictated by an adversary, but typically subject to constraints on the rate of packet injection and how closely in time (how bursty) these arrivals may be scheduled. Even more challenging, there is a growing literature addressing the case where packet arrival times are set by an unconstrained adversary; see \cite{DeMarco:2017:ASC:3087801.3087831,bender2020contention,fineman:contention2}.  Recent work in this area addresses additional challenges facing modern networks, such as energy efficiency \cite{jurdzinski:energy,bender:contention}, malicious disruption of the shared channel \cite{bender:how,awerbuch:jamming,ogierman:competitive,DBLP:conf/podc/ChenJZ21,anantharamu2019packet,anantharamu2011medium}, and the ability to have a limited number of concurrent transmission succeed \cite{bender:coded}.

\section{Technical Overview for Upper Bound}\label{sec:technical-overview}

\noindent Due to space constraints, our proofs for the upper bound are provided in Section \ref{analysis-section} of our appendix. However, we do reference some well-known  inequalities based on the Taylor series that are omitted here, but  can be found in Section \ref{sec:prelims}. 

Here, we present an overview of our analysis, with the aim of imparting some intuition the design choices behind \malg, as well as  highlighting the novelty of our approach.  To this end, we first consider two natural---but ultimately flawed---ideas for solving our problem.\medskip 

\noindent{\bf Straw Man 1.} An immediate question is: {\it Why can we not use a prior contention-resolution algorithm to solve our problem?} For example, in the static setting, a well-known backoff algorithm, such \beb guarantees \whp that all packets succeed with makespan $\Theta(n \log n)$ \cite{bender:adversarial}. Under BEB, the packets execute over a sequence of disjoint windows, where window $i\geq 0$ consists of $2^i$ contiguous slots. Every active packet sends in a slot chosen \uar from the current window. Unfortunately, a constant fraction of slots in each window $i\leq \lg(n) + O(1)$ will be collisions. Each collision imposes a cost of $\CollCost$ leading to a collision cost of $\Omega(n\CollCost)$. \qed
\medskip

\noindent Despite yielding a poor result, this straw man provides three useful insights. First, we cannot have packets start with a ``small'' window, since this leads to many collisions. Many backoff algorithms start with a small window, and will not yield good performance for this same reason.

Second,  we should seek to better (asymptotically) balance the costs of makespan and collisions. Under BEB, these costs are highly unbalanced, being $\Theta(n\log n)$ and $\Omega(n\CollCost)$, respectively. We may trade off between makespan and collision cost; that is, we can make our windows larger, which increases our makespan, in order to dilute the probability of collisions. 

Third, a window of size $\Theta(n\sqrt{\CollCost})$ seems to align with these first two insights; that is, it appears to asymptotically balance makspan and collision cost. To understand the latter claim, suppose that in each slot,  every packet sends with probability $\Theta(1/(n\sqrt{\CollCost}))$. We can argue (informally) that, for any fixed slot, the probability of any two packets colliding is at least $\binom{n}{2}\Theta(1/n\sqrt{\CollCost})^2 = O(1/\CollCost)$. Thus, over $n\sqrt{\CollCost}$ slots in the window, the expected number of collisions is  $O( n\sqrt{\CollCost}/\CollCost) = O(n/\sqrt{\CollCost})$, and each collision has cost $\CollCost$, so the expected collision cost is $O(n\sqrt{\CollCost})$. Of course, this informal analysis falls short, since collisions may involve more than two packets, and we have not shown that all $n$ packets can succeed over this single window (they cannot). Yet, this insight offers us hope that we can outperform prior backoff algorithms  by achieving costs that are $o(n\CollCost)$.

How should we find a window of $\Theta(n\sqrt{\CollCost})$? Since $n$ is unknown {\it a priori}, we cannot simply instruct packets to start with that window size. It is also clear from the above discussion that we cannot grow the window to this size using prior backoff algorithms. This obstacle leads us to our second straw man.\medskip 

\noindent{\bf Straw Man 2.} Perhaps we can estimate $n$ and then start directly with a window of size $\Theta(n\sqrt{\CollCost})$. A well-known ``folklore'' algorithm for estimating $n$ is the following. In each slot $i\geq 0$, each packet sends  with probability $1/2^i$; otherwise, the packet listens. This algorithm, along with improvements to the quality of the estimation, is explored by Jurdzinski et al. \cite{jurdzinski:energy}, but we sketch why it works here.

Intuitively, when $i$ is small, say a constant, then the probability of an empty slot is very small: $(1-1/2^{\Theta(1)})^n \leq e^{-\Theta(n)}$ by Fact \ref{fact:taylor}(a). However, once $i=\lg(n)$, then the probability of an empty slot is a constant: $(1 - 1/2^{\lg(n)})^{n} = (1 -1/n)^{n} \geq e^{\Theta(1)}$ by Fact \ref{fact:taylor}(c). In other words, once a packet witnesses an empty slot, it can infer that its sending probability is $\Theta(1/n)$, and thus the reciprocal yields an estimate of $n$ to within a constant factor. 

At first glance, it may seem that we can estimate $n$, and then proceed to send packets in a window of size $\Theta(n\sqrt{\CollCost})$. Unfortunately, for each slot $i\leq \lg(n)$, this algorithm (and others like it) will likely  incur a collision, and thus the expected collision cost is $\Omega(\CollCost \log n)$.   To see why this is a problem, suppose that $\CollCost = n^4$. This implies that the expected collision cost is $\tilde{\Omega}(\CollCost) = \tilde{\Omega}(n^4)$, while the makespan is at best $O(n\sqrt{\CollCost}) = O(n\sqrt{n^4}) = O(n^3)$. That is, there is an $n$-factor discrepancy between the expected collision cost and makespan, which grows worse for larger values of $\CollCost$;  this approach is not trading off well between these two metrics.\qed
\medskip

Even though the above algorithm racks up a large collision cost, it highlights how channel feedback can provide useful hints. If packets are less aggressive with their sending, we can reduce collisions while still receiving feedback that lets us reach our desired window size of $\Theta(n\sqrt{\CollCost})$.  

\subsection{Our Approach} \label{sec:our-approach}

%The design of our algorithm, \mainAlgorithm, is informed by the above discussion, with the aim of obtaining bound of $\tilde{O}(n\sqrt{\CollCost})$ for makespan and the expected collision cost, which yields an expected completion time of $\tilde{O}(n\sqrt{\CollCost})$. 

%As with many of the original backoff algorithms, we find it useful to reason about \malg using the notion of a current window size, {\boldmath{$w$}}. This 

%In the current slot, we denote the current window size by {\boldmath{$\currWin$}}, and each active packet sends with probability $\Theta(1/\currWin)$. 

\noindent{\bf Motivating Sample Size.} Hoping to avoid the problems illustrated in our discussion above, all packets begin with an initial current window size that can be ``large''; the exact size we motivate momentarily.  Recall that we aim for packets to tune their sending probability to be $\Theta(1/(n\sqrt{\CollCost}))$, so we are aiming for a window size of $\Theta(n\sqrt{\CollCost})$. However,  this must be achieved with low expected collision cost. To do this, the packets execute over a \defn{sample} of {\boldmath{$s$}} slots. For each slot in the sample, every packet sends with probability $1$ divided by the current window size, and monitors the channel feedback.  

What does this channel feedback from the sample tell us? Suppose we have reached our desired window size of $\Theta(n\sqrt{\CollCost})$. Then, the window size may be large relative to $n$, and so we should expect empty slots to occur frequently, while successes {\it do} occur, but less often; the probability of success is approximately $\binom{n}{1}(1/(n\sqrt{\CollCost})) \approx 1/\sqrt{\CollCost}$. To correctly diagnose \whp (using a Chernoff bound) that $w = \Theta(n\sqrt{\CollCost})$, we rely on receiving $\Theta(\log n)$ successes. This dictates our sample size to be $s  = \Theta(\sqrt{\CollCost}\log(n\sqrt{\CollCost})) = \Theta(\sqrt{\CollCost}\log(n))$, since $\CollCost \leq n^{\kappa}$.

The details behind this intuition are given in Section \ref{sec:diagnosis-correctness} of our appendix, where we show that samples of size $\Theta(\sqrt{\CollCost}\log(n))$ are sufficient to make correct decisions that tune the window size to be $\Theta(n\sqrt{\CollCost})$. Furthermore, in Section \ref{sec:sample-cost} of our appendix, we show that the corresponding cost from this tuning is  $\tilde{O}(n\sqrt{\CollCost})$.
\medskip

\noindent{\bf Motivating Initial Window Size.} Suppose we are not at our desired window size. Then, the feedback from sampling tells us what to do. If our current window size is too  large (i.e., our sending probability is too low), then the number of successes will be ``small'', since most slots are empty, and we should decrease our window size.  Else, if the current window size is too small, (i.e., our sending probability is too high)  then the number of successes is again ``small'', since  most slots are collisions, and we should increase the window size. But wait, in the latter case, can we tolerate the cost of the resulting collisions? In the worse case, the entire sample might consist of collisions, resulting in a potentially large cost of $s\,\CollCost  = \Theta(\CollCost^{3/2}\log n)$.

We remedy this problem by setting our initial window size to be $\CollCost$. Then, the probability of a collision in a fixed slot is approximately $\binom{n}{2}(1/\CollCost^2)$ and so the expected cost is roughly $T = \binom{n}{2}(1/\CollCost^2) \,\CollCost = \Theta(n^2/\CollCost)$. Reasonaing informally, if $\CollCost > n$, then $T = O(n)$, and over the sample the expected cost is $O(n\sqrt{\CollCost}\log (n))$, as desired.
Else, if $\CollCost \leq n$, then even if our sample does consist entirely of collisions, the resulting cost is $\Theta(\CollCost \sqrt{\CollCost} \log n) = O(n\sqrt{\CollCost}\log n)$. This reasoning is formalized in Lemma \ref{l:okay-TH} of our appendix.

In light of the above, we emphasize that \malg does not avoid all collisions; in fact, many collisions may occur when the collision cost is ``small'' (i.e., $\CollCost \leq n$). However, as the collision cost grows ``larger'' (i.e., $\CollCost > n$), \malg expresses an increasing aversion to collisions by having packets tune their respective sending probabilities such that we expect $o(1)$ collisions per sample. The cost analysis for sampling is given in Section \ref{sec:sample-cost} of our appendix.

\medskip

\noindent{\bf Borrowing from Sawtooth Backoff.} By using sampling, ultimately a window size of $\Theta(n\sqrt{\CollCost})$ is reached. At this point, every active packet executes the analog of the final run of \stb (recall Section \ref{sec:related}). Specifically, each packet sends with probability $p=\Theta(1/(n\sqrt{\CollCost}))$, which we show is sufficient to have at least  half of the active packets succeed (see Lemma \ref{l:single-rundown}). Then, the window is halved, and the process repeats where the remaining packets send with probability $p/2$. This halving process continues until all packets succeed. By the sum of a geometric series, the  number of slots in this process is $O(n\sqrt{\CollCost})$. Informally, the expected collision cost per slot is roughly $\binom{n}{2}(1/(n\sqrt{\CollCost})^2)\CollCost = O(1)$, and thus  $O(n\sqrt{\CollCost})$ over all slots in the window. The analysis of cost is provided in Section \ref{sec:rundownfull} of our appendix.

\begin{figure}[t]
%\begin{tcolorbox}[standard jigsaw, opacityback=0]
\includegraphics[width=\textwidth]{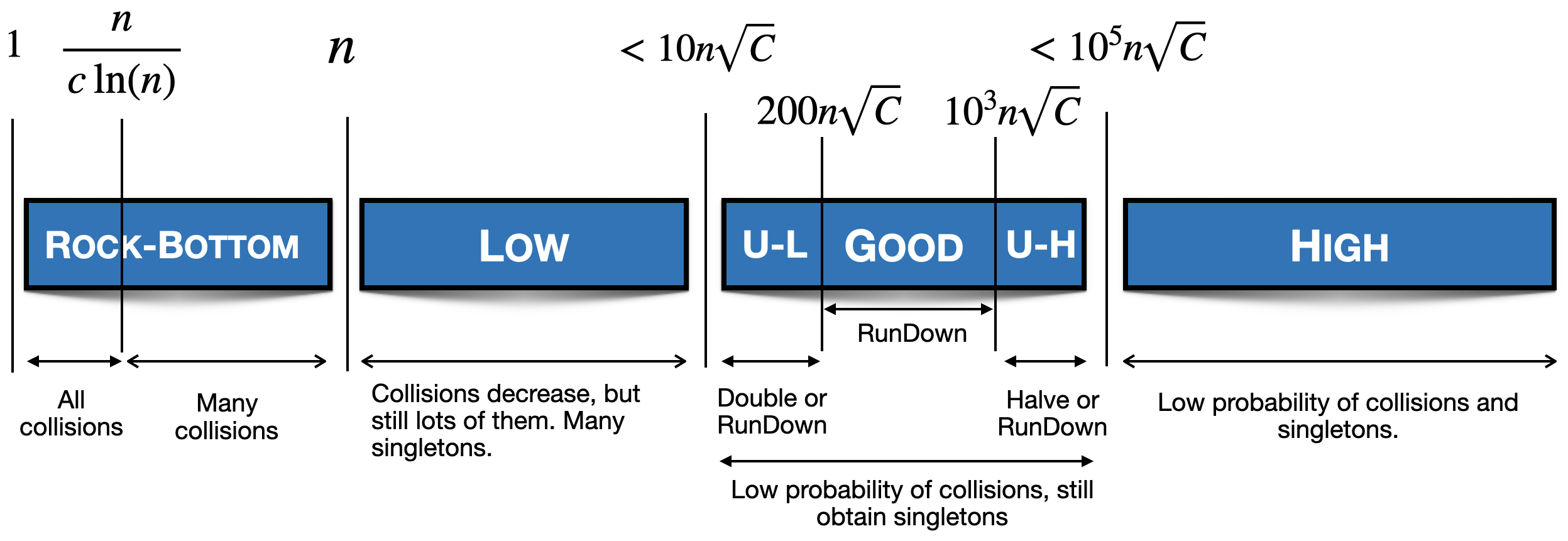}
%\end{tcolorbox}
\vspace{-15pt}\caption{Illustration of the ranges discussed in Section \ref{sec:algorithm}.}
\label{algo-range-visualise}
\end{figure}

\subsection{Our Algorithm and Overview of Analysis}\label{sec:algorithm}
\smallskip
This section describes and gives intuition for \malg, whose pseudocode is given in Figure~\ref{our-algoritm}. As stated in our model (Section \ref{sec:model}), when a packet succeeds it terminates immediate; for ease of presentation, we omit this from our pseudocode.

From a high-level view, each packet keeps track of a current window size, {\boldmath{$\currWin$}}; this size is critical, as it dictates the per-slot sending probability of each packet, which is $\Theta(1/\currWin)$. Each packet keeps track of its own notion of a current window. However, in each slot, since every packet is either listening or sending (and, thus, learning whether the slot contained a success or a collision), all packets receive the same channel feedback and adjust their respective current window identically (stated formally in Lemma~\ref{l:same-view}). Therefore, for ease of presentation, we refer only to a single current window.\medskip

%As discussed above, packets receive channel feedback, and this feedback is used to increase or decrease $\currWin$.  %Once $\currWin$ reaches a ``good'' range (described below), the packets will start terminating when they succeed.

\noindent{\bf Defining Ranges.} In order to describe how \malg works, we define the six size \defn{ranges} that the current window (or just ``window'' for short), $\currWin$, can belong to during an execution.

\renewcommand{\labelitemi}{$\bullet$}
\begin{itemize}
\setlength\itemsep{-0pt}
\item \tootoolow: $[1, n)$.
\item \toolow: $[n, 10 n\sqrt{\CollCost})$.
\item \uncertainA: $[10 n\sqrt{\CollCost},200n\sqrt{\CollCost}) $.
\item \goodrange: $[200n\sqrt{\CollCost}, 10^3 n\sqrt{\CollCost})$.
\item \uncertainB: $[10^3 n\sqrt{\CollCost},10^5 n\sqrt{\CollCost}) $.
\item \toohigh: $[ 10^5 n\sqrt{\CollCost}, \infty)$.
\end{itemize}

These ranges are depicted in Figure \ref{algo-range-visualise}. We note that the particular constants in these ranges are not special; they are chosen for ease of analysis. To gain intuition, we now describe the events we expect to witness in the \tootoolow, \toolow, \goodrange, and \toohigh ranges when \malg executes; we defer an in-depth discussion of \uncertainA and \uncertainB until the end of the next subsection.

The \tootoolow range captures ``tiny'' window sizes, starting from $1$ up to $n-1$. In this range, the probability of sending exceeds $1/n$, and we expect that most slots will be collisions. The next range is \toolow, and it includes window sizes from $n$ to just below $10n\sqrt{\CollCost}$. This range represents a moderate increase in window size, allowing for more successes than \tootoolow, although there can still be many collisions. As discussed in Section \ref{sec:our-approach}, in \tootoolow and in the bottom portion of \toolow, we can afford to have an single sample consist entirely of collisions, since $\CollCost = O(n)$ (see Lemma \ref{lm-tl-exp-collisions}). However, lingering in these ranges would ultimately lead to sub-optimal expected collision cost.

The \goodrange range spans from $200n\sqrt{\CollCost}$ to just below $10^3n\sqrt{\CollCost}$. In this range, the window sizes are sufficiently large that we expect $o(1)$ collisions, along with handful of successes; notably, less successes than \toolow. This turns out to be a ``good'' operating range for the algorithm, where the balance between collision costs and makespan is achieved, as discussed in Section \ref{sec:technical-overview} (i.e., our third insight after Straw Man 1).

The \toohigh range covers window sizes from $10^5n\sqrt{\CollCost}$ and above. In this range, the window sizes are very large, leading to a very low probability of collisions, which is good. However, there are also far-fewer successes compared to \goodrange, which means that lingering in this range would lead to a sub-optimal number of slots until all packets succeed.

%%%%%%%%%%%%%%%%%%%%%%%%%%%%%%%%%%%%%%%%%%%%%%%
%%%%%%%%%%%%%%%%%%%%%%%%%%%%%%%%%%%%%%%%%%%%%%%
%%%%%%%%%%%%%%%%%%%%%%%%%%%%%%%%%%%%%%%%%%%%%%%
{\renewcommand{\baselinestretch}{1}
\tiny

\begin{figure}[ht!]
{\fontsize{9.5}{12}\selectfont % shrinking font
\begin{algorithm}[H]
\caption{\bf \mainAlgorithm}%\label{our-algoritm}

\SetKwFunction{FMain}{CollectSample}
\SetKwFunction{FDiag}{Diagnosis}
\SetKwFunction{FRun}{RunDown}
\SetKwFunction{FRunFull}{RunDownFull}
\SetKwProg{Fn}{Function}{:}{}

\BlankLine

Initial window size $w_{cur} \leftarrow \CollCost$ \label{alg:first-line}\\
%At all times, sending probability $p \leftarrow \frac{1}{\currWin}$ \\
%At all times, Sample size $\leftarrow d\sqrt{\CollCost} \ln (w_{cur})$ \\

\medskip
\Repeat{true}{
 {\it \#successes, \#collisions}  $\leftarrow 0$\\
           \Sample{($\CollCost,\currWin$)}\\
           \Diagnosis{(\#successes, \#collisions, $\currWin$)} \\
        
        }  

\medskip

%\SetKwFunction{FMain}{CollectSample}{\label{alg:fun-CollectSample}}
    %\SetKwProg{Fn}{Function}{:}{}
     \Fn{\Sample{($\CollCost,\currWin$)}}{\label{alg:fun-CollectSample}
    {\it \#successes} $\leftarrow 0$\\
{\it \#collision} $\leftarrow 0$\\
        \For{slot $i = 1$ \KwTo $d\sqrt{\CollCost}\ln(\currWin)$}{
            Send with probability $1/\currWin$; otherwise, listen\\
            \If{slot is a success}
                {\it \#successes++}
                \ElseIf{slot is a collision}
                {\it \#collisions++}          
        }
    }

\medskip    
%\SetKwFunction{FMain}{Diagnosis}{\label{alg:fun-Diagnosis}}
%\SetKwFunction{FRunDown}{RunDown}
    %\SetKwProg{Fn}{Function}{:}{}
    \Fn{\Diagnosis{(\#successes, \#collisions, $\currWin$)}}{\label{alg:fun-Diagnosis}    
\If{\#successes $>\frac{2d\ln(\currWin)}{10^5}$}{ \label{alg:first-if}
        \If{\#successes $\leq \frac{d\ln(\currWin)}{20e}$}{\label{alg:second-if}
            \If{\#collisions $\geq \frac{d\sqrt{\CollCost}\ln(\currWin)}{8e^2}$}{\label{alg:if-if-if}
                $\currWin \leftarrow 2\currWin$ \label{alg:double-the-first}\\
            }
            \Else{\label{alg:execute-rundownfull}
                 %{\it \#voteRunDown} 
                 Execute \Rundown($\currWin$)\\
            }
        }
        \Else{ \label{alg:second-else} $\currWin \leftarrow 2\currWin$} \label{alg:double-blah}
    }
    \Else{\label{alg:first-else}
        \If{\#collisions $\geq \frac{d\sqrt{\CollCost}\ln(\currWin)}{8e^2}$}{\label{alg:third-else}
            %{\it \#voteDoubleWindow}  
            $\currWin \leftarrow 2\currWin$\\
        }
        \Else{\label{alg:halve-window}
           %{\it \#voteHalveWindow}
           $\currWin \leftarrow \currWin/2$\\
        }
    }
}

\medskip

\Fn{\Rundown{($\currWin$)}}{\label{alg:fun-RunDown}
    $w_0 \leftarrow \currWin$\\
    \While{$\currWin \geq 8\sqrt{\CollCost}\lg(w_0)$}{\label{alg:while}
        \For{each slot $j = 1$ \KwTo $\currWin$}{\label{alg:run-start}
        Send packet with probability $2/\currWin$ \label{alg:rsingle-window-rundown}
    }
    %All packets that succeeded terminate\\
    %If the packet succeeded, then terminate\\
    $\currWin\leftarrow \currWin/2$\label{alg:run-end}\\
  }
  
  \For{$i = 1$ \KwTo $c\ln(w_0)$ \label{alg:repeat-runs}}{
        \For{each slot $j = 1$ \KwTo $w_0$}{\label{alg:repeat-run-start}
        Send packet with probability $2/w_0$\label{alg:repeat-rsingle-window-rundown}
        }
    }
} 
\end{algorithm}
\caption{Pseudocode for \mainAlgorithm.}\label{our-algoritm}
} % end of selectfont
\end{figure}
}
%%%%%%%%%%%%%%%%%%%%%%%%%%%%%%%%%%%%%%%%%%%%%%%
%%%%%%%%%%%%%%%%%%%%%%%%%%%%%%%%%%%%%%%%%%%%%%%
%%%%%%%%%%%%%%%%%%%%%%%%%%%%%%%%%%%%%%%%%%%%%%%

\clearpage

The  \uncertainA and \uncertainB ranges span $[10n\sqrt{\CollCost}, 200n\sqrt{\CollCost})$ and $[10^3n$ $\sqrt{\CollCost},$ $10^5n\sqrt{\CollCost})$, respectively. Why do we have these ranges? They capture values of $\currWin$ where we cannot know \whp exactly what will happen, although we {\it do} prove that the algorithm is making ``progress''. We will elaborate on this further after describing our methods for sampling and interpreting channel feedback, which we address next.\medskip

\noindent{\bf Sampling and Diagnosing Feedback.}  \malg navigates these ranges by using two subroutines: {\bf \Sample} and {\bf \Diagnosis}. As motivated earlier  in Section \ref{sec:our-approach}, a sample is a contiguous set of $d\sqrt{\CollCost}\ln(\currWin)$ slots that are used by each packet to collected channel feedback. Specifically, under \Sample, in each slot of the sample, every packet sends independently with probability $1/\currWin$ and records the result if it is a success or a collision.

Based on the number of successes and the number of collisions, \Diagnosis   attempts to determine the range to which $\currWin$ belongs. We discuss these thresholds in order to give some intuition for how \Diagnosis is making this determination. To simplify the presentation, we omit discussion of \uncertainA and \uncertainB until the end of this section. The process by which \malg homes in on the range to which $\currWin$ belongs is depicted in Figure \ref{algo-regime-diff-box}.

We begin with the first if-statement on Line \ref{alg:first-if}. This line checks if the number of successes is at least a logarithmic amount (in $\currWin$); if so, this indicates that $\currWin$ cannot be in the \toohigh range, which would yield far fewer successes (see Lemma \ref{l:okay-TH}). Therefore, if meet the conditional on Line \ref{alg:first-if}, it must be the case that $\currWin$ belongs to \tootoolow, \toolow, or \goodrange.

Line \ref{alg:second-if} checks whether the number of successes falls below a ``large'' logarithmic amount; if not, then we are seeing a ``large'' number of successes that indicates $\currWin$ cannot be in \tootoolow or in \goodrange, and so must be in \toolow  (see Lemma~\ref{l:okay-TL}). Otherwise, the number of successes falls below this logarithmic amount, indicating  that $\currWin$ is in the \tootoolow or in \goodrange ranges.  To discern between these two cases, Line \ref{alg:if-if-if} checks if the number of collisions exceeds $\frac{d\sqrt{\CollCost}\ln(\currWin)}{8e^2}$, which indicates $\currWin\in\tootoolow$ (see Lemmas
\ref{l:all-collisions-if-smallW}, \ref{l:upper-too-too-low}, and \ref{l:okay-TTL}). Otherwise, $\currWin$ falls within the \goodrange and the job of \Diagnosis is complete --- at this point,  \Rundown  will be executed.

How can an algorithm with low expected collision cost make a decision based on whether $\Theta(\sqrt{\CollCost} \log(\currWin))$ collisions occur? Recall that in this case, we are deciding between $\currWin \in \tootoolow$ and $\currWin \in \goodrange$. We will only witness this number of collisions when $\currWin$ is in the \tootoolow range, where $\CollCost \leq n$ and so we can tolerate this cost. Otherwise, as described above, $\currWin \in \goodrange$ and the expected number of collisions is only $\tilde{O}(1/\sqrt{\CollCost})$ (see the discussion preceding Lemma \ref{l:sizable-window-upper-collisions} in our appendix).

\begin{figure}[t]
\begin{tcolorbox}[standard jigsaw, opacityback=0]
\includegraphics[width=0.98\textwidth]{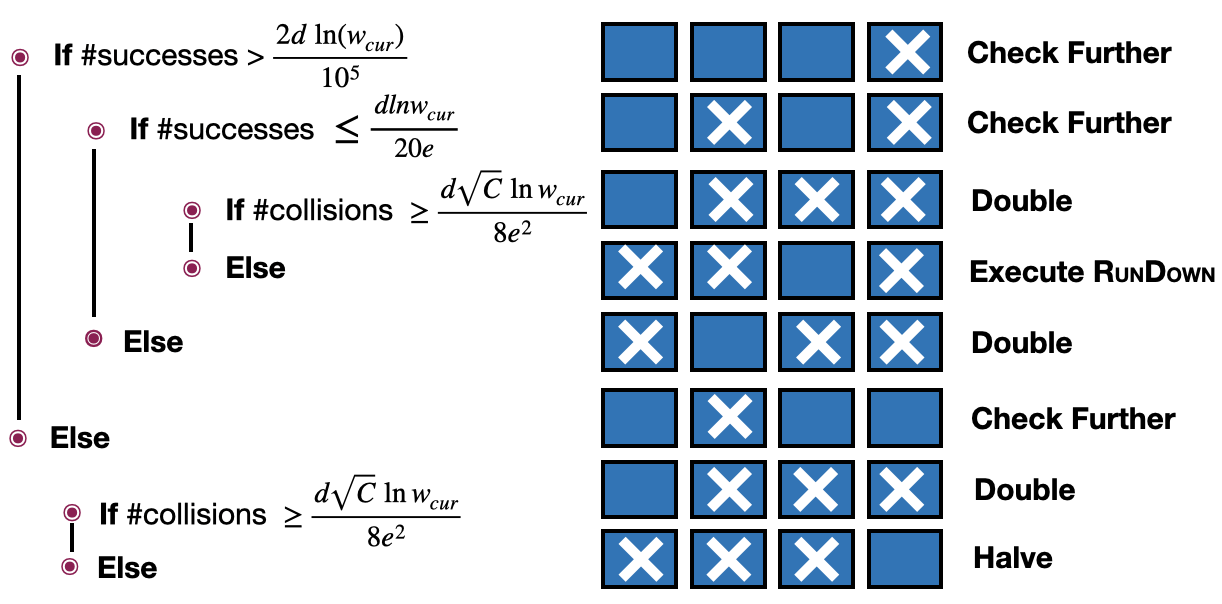}
\end{tcolorbox}
\vspace{-10pt}\caption{The elimination of ranges as discussed in Section \ref{sec:algorithm}.}
\label{algo-regime-diff-box}
\vspace{-15pt}
\end{figure}

The remaining portion of \Diagnosis starts with the else-statement on Line \ref{alg:first-else}. This line is executed only if the conditional on Line \ref{alg:first-if} is not met; that is, we have very few successes. This can occur only if $\currWin\in\tootoolow$ or $\currWin\in\toohigh$. The former case is (again) diagnosed be checking whether there are many collisions (Line \ref{alg:third-else}) and, if so, the window is doubled; otherwise, we are in the latter case and the window is halved. Again, we only have many collisions if  $\currWin$ is in the \tootoolow or lower portion of the \toolow ranges, where $\CollCost \leq n$ and so we can tolerate this cost.

What about the uncertain ranges? In \uncertainA, a sample will contain enough successes to satisfy Line \ref{alg:first-if}. However, it is unclear whether the number of successes will be ``large'' (if $\currWin$ is in the lower portion of \uncertainA) or ``moderate'' (if $\currWin$ is in the upper portion of \uncertainA). In the former case, we fail Line \ref{alg:second-if}, while in the latter case, we satisfy Line \ref{alg:second-if}. Despite this uncertainty, observe that we are nonetheless guaranteed to either double the window or execute \Rundown. Either of these outcomes count as progress, since either $\currWin$ moves closer to \goodrange by doubling, or \Rundown is executed on a window of size  $\Theta(n\sqrt{\CollCost})$ (in contrast, halving $\currWin$ would be counterproductive).

The intuition behind \uncertainB is similar. If $\currWin \in \toohigh$, then we can show that the number of successes fails Line \ref{alg:first-if} (see Lemma~\ref{l:okay-TH}) and that, ultimately, we halve $\currWin$. However, between \goodrange and \toohigh, this cannot be shown \whp; it may hold if $\currWin$ is close to the lower end of \toohigh, but not hold if $\currWin$ is close to the upper end of \goodrange. So, instead, we argue that we either halve $\currWin$ or execute \Rundown. Either of these actions counts as progress, since either $\currWin$ moves closer to \goodrange, or \Rundown is executed on a window of size $\Theta(n\sqrt{\CollCost})$. 

 %In the former case,  $\currWin$ $\in$ $[10 n\sqrt{\CollCost},$ $200 n\sqrt{\CollCost})$, which is a range where  the number of successes will exceed $\frac{2d\ln(\currWin)}{10^5}$, so the if-statement on Line \ref{alg:first-if} will be satisfied. At this point, by specification of \Diagnosis, the only outcomes are that $\currWin$ doubles, or \Rundown is executed (see Lemma~\ref{l:limbo-low}). If the window doubles, then we are making progress towards \goodrange (where \Rundown will be executed); otherwise, we execute \Rundown immediately, which is fine, since $\currWin = \Theta(n\sqrt{\CollCost})$. 

%When  $\currWin\in$\uncertainB, we cannot guarantee that Line \ref{alg:first-if} will be satisfied. If so, then \Diagnosis proceeds just as in the case of $\currWin\in\goodrange$ (and will execute \Rundown). Else, we move to Line \ref{alg:first-else} and show that, since there are very few collisions, it must be the case that the window is halved via Line \ref{alg:halve-window}.\medskip

 Ultimately, we are able to show the following key lemma (in Section \ref{sec:sample-cost} of our appendix):\medskip

\noindent{\bf Lemma \ref{lem:correctness-cost-collectsample-diagnosis}.} {\it 
The executions of \Sample and \Diagnosis guarantee that \whp: (i) \Rundown is executed within  $O(\sqrt{\CollCost}\log^2(n) )$ slots,  and (ii) the total expected collision cost until that point is $O(n\sqrt{\CollCost} \log^2(n))$.
}
\medskip

\noindent{\bf All Remaining Active Packets Succeed.} The final subroutine,  {\bf \Rundown}, is executed once $\currWin=\Theta(n\sqrt{\CollCost})$, and it allows all active packets to succeed. In each slot of $\currWin$, every packet sends with probability $2/\currWin$. Otherwise, the current window size is halved  and the remaining active packets repeat this process. This continues until the window size reaches  $\Theta(\sqrt{\CollCost} \log(w_0)) =\Theta(\sqrt{\CollCost} \log(n))$, where the asymptotic equality holds by recalling that $\CollCost=\poly(n)$. Once this smallest window in the run is reached,  $O(\log n)$ active packets remain. To finish these packets, \malg  performs an additional $\Theta(\ln(n\sqrt{\CollCost}))=\Theta(\ln(n))$ windows of size $\Theta(n\sqrt{\CollCost})$,  where any remaining active packet sends in each slot with probability $\Theta(1/n\sqrt{\CollCost})$. 

We prove the following lemmas (in Section \ref{sec:rundownfull} of our appendix):
\medskip

\noindent{\bf Lemma \ref{lem:correctness-latency-rundown}.} {\it 
When \Rundown is executed, \whp all packets succeed within $O(n\sqrt{\CollCost}\ln(n))$ slots.}\medskip

\noindent{\bf Lemma \ref{lem:expected-cost-collision-rundown}.} {\it
W.h.p. the expected collision cost for executing \Rundown is $O(n\sqrt{\CollCost}\ln(n))$.
}\medskip

\noindent Finally, our upper bound in Theorem \ref{thm:main-theorem} follows directly from Lemmas \ref{lem:correctness-cost-collectsample-diagnosis}, \ref{lem:correctness-latency-rundown}, and \ref{lem:expected-cost-collision-rundown}.

\section{Technical Overview for Lower Bound}\label{sec:lower-bound}

In this section,  we provide an overview of our argument, which focuses on placing a lower bound on the expected collision cost. Here, we highlight the key lemmas in our argument, and our full proofs are provided in Section \ref{app:lowerbound} of our appendix.  

%Our argument addresses two main cases: (1) algorithms where contention is ``high'' for at least one slot, and (2) algorithms where contention is ``low'' in all slots. 

We consider only the set of slots, {\boldmath{$\twoActiveSlots$}}, in the execution of any algorithm where at least two active packets remain, since we cannot have a collision with a single active packet. While we do not always make this explicit, but going forward, any slot $t$ is assumed to implicitly belong to $\twoActiveSlots$.

Let {\boldmath{$p_i(t)$}} denote the probability that packet $i$ sends in slot $t$.  Note that, if a packet has terminated, its sending probability can be viewed as $0$. For any fixed slot $t$, the \defn{contention} in slot $t$ is  {\bf \texttt{Con}}{\boldmath{$(t)$}} $= \sum_{i=1}^n p_i(t)$; that is, the sum of the sending probabilities in that slot.
\medskip

\noindent{{\underline{\bf When Contention is High}}.}  We start by showing that any algorithm that has even a single slot $t$ with $\con > 2$ must have $\Omega(\CollCost)$ expected collision cost, which is one portion of our lower bound. This is done by deriving an expression for the probability of a collision in any fixed slot $t$ as a function of $\con$, which is useful  when  $\con$ is ``high'' (see Lemmas \ref{lem:prob-success}, \ref{lem:high-contention}, and \ref{lem:omega-collision} in Section \ref{app:lowerbound} of our appendix).\medskip

\noindent{{\underline{\bf When Contention is Low}}.} What about an algorithm where all slots of the execution have ``low'' contention (i.e., $\con \leq 2$)?  Our previous expression is hard to work with in this case. So, we derive a different expression for the probability of a collision in a fixed slot $t$, which can be useful for small values of $\con$:\medskip

\noindent{\bf Lemma \ref{lem:lower-prob-coll}.} {\it 
Fix any slot $t$ and let $\con\leq 2$. The probability of a collision in $t$ is at least $(\frac{1}{110})\left(\con^2 - \sum_i p_i(t)^2 \right)$.}\medskip

\noindent However, this expression requires some additional work to be deployed in our argument, as we now describe.  

For any fixed slot $t\in \twoActiveSlots$, let {\boldmath{$\pmax(t)$}} be the maximum sending probability of any packet in slot $t$. Similarly, let {\boldmath{$\psec(t)$}} $\leq \pmax(t)$  be the next-largest sending probability of any packet in slot $t$; note that $\psec(t)=\pmax(t)$, if more than one packet sends with probability $\pmax(t)$.  

Define {\boldmath{$\Delta(t)$}} $= \psec(t)/\pmax(t)$.  Our analysis  ignores any slot $t$ where $\con = 0$; that is, any slot where every packet has a sending probability of $0$. We give such slots to any algorithm for ``free''; that is, we do not include them in the cost of the execution. Thus, $\Delta(t)$ is always well-defined, since $\pmax(t)>0$. 

Why do we need $\Delta(t)$? It captures a sense of ``balance''. For our purposes, the situation is most ``unbalanced'' when exactly one packet has non-zero sending probability, while all other packets have zero sending probability; that is, when the total value of $\con>0$ is due to a single packet. Clearly, in such slots, there can be no collision and, correspondingly, $\Delta(t)=0$. In our argument,  $\Delta(t)$ plays a key role in establishing the following:\medskip

\noindent{\bf Lemma \ref{lem:con-minus-squares}.} {\it 
For any fixed slot $t$,  $\con^2 - \sum_i p_i(t)^2 \geq  \Delta(t)\, \con^2/2$.}\medskip

\noindent A natural extension of $\Delta(t)$ is {\boldmath{$\deltaMin$}}, which is the minimum $\Delta(t)$ over all slots $t\in\twoActiveSlots$. As we show momentarily, our lower bound is parameterized by $\deltaMin$. The last component of our argument addresses the sum of the contention over all slots in $\twoActiveSlots$:\medskip

\noindent{\bf Lemma \ref{lem:sum-contention-restricted}.} {\it 
Any algorithm that guarantees \whp that $n$ packets succeed must \whp have \\ $\sum_{t\in\twoActiveSlots} \con$ $= \Omega(n)$.} \medskip

\noindent We can now establish a lower bound for this low-contention case:\medskip

\noindent{\bf Lemma \ref{lem:lower-bound}.} {\it 
Consider any algorithm $\mathcal{A}$ whose contention in any slot is at most $2$ and guarantees \whp a makespan of $\tilde{O}(n\sqrt{\CollCost})$.  W.h.p. the expected collision cost for $\mathcal{A}$ is $\tilde{\Omega}(\deltaMin n\sqrt{\CollCost})$.}
\begin{proof}
Let $X$ be a random variable that is $|\twoActiveSlots|$ under the execution of $\mathcal{A}$. Note that, since \whp the makespan is $O(n\sqrt{\CollCost})$, it must be the case that \whp $X = O(n\sqrt{\CollCost})$. 

By Lemma \ref{lem:sum-contention-restricted}, we have:
\begin{align}
    \sum_{t=1}^{X}  \con \geq c n. \label{eq:contention-sum}
\end{align}
\noindent for some constant $c>0$. Let $Y_{t}=1$ if slot $t \in \twoActiveSlots$  has a collision; otherwise, $Y_{t}=0$. By Lemmas \ref{lem:lower-prob-coll} and \ref{lem:con-minus-squares}:
\begin{align*}
 Pr(Y_t=1) & \geq   \Delta(t) \cdot \con^2/220 \nonumber\\
 & \geq \deltaMin \cdot \con^2/220 
\end{align*}
where the second line  follows from the definition of $\deltaMin$. The expected collision cost is:
\begin{align}
     \sum_{t=1}^X P(Y_t=1)\cdot \CollCost & \geq \sum_{t=1}^X \frac{\deltaMin \con^2 \cdot \CollCost} {220} \nonumber\\
     & = \frac{\deltaMin \cdot \CollCost}{220} \sum_{t=1}^X \con^2. \label{eq:collision-sum}
\end{align}
By Jensen's inequality for convex functions, we have:
\begin{align}
    \frac{\sum_{t=1}^X \con^2}{X} \geq \left(\frac{\sum_{t=1}^X \con}{X}\right)^2 \label{eq:jensen}
\end{align}  
Finally, the expected cost is at least:
\begin{align*}
    \frac{\deltaMin \cdot \CollCost}{220} \sum_{t=1}^X \con^2 & \geq   \frac{\deltaMin \cdot \CollCost}{220} \frac{\left(\sum_{t=1}^X \con\right)^2}{X} \mbox{~~~by Equations \ref{eq:collision-sum} and \ref{eq:jensen}}\\
    & \geq \frac{\deltaMin \cdot \CollCost}{220} \left(\frac{c^2 n^2}{X}\right) \mbox{~~~by Equation \ref{eq:contention-sum}}\\
    & = \tilde{\Omega}\left( \deltaMin\, n\sqrt{\CollCost}  \right)
\end{align*}
\noindent where the second line follows by Equation \ref{eq:contention-sum},  which was defined with regard to $\twoActiveSlots$, and so can be compared to Equation  \ref{eq:collision-sum}. The last line follows since \whp $X= \tilde{O}(n\sqrt{\CollCost})$.\qed
\end{proof}

\noindent Given our analysis of the high- and low-contention cases, we have the following  lower bound:\medskip

%\begin{theorem}\label{thm:general-lower-bound-main}
\noindent{\bf Theorem \ref{thm:general-lower-bound-main}.} {\it 
Consider any algorithm that \whp guarantees $\tilde{O}(n\sqrt{\CollCost})$ makespan. Then, \whp, the expected collision cost for $\mathcal{A}$ is $\tilde{\Omega}(\min\{\CollCost, \deltaMin n\sqrt{\CollCost}\})$.}    
%\end{theorem}

\medskip

What algorithms does our lower bound say something interesting about? We can start with our statement in Theorem \ref{thm:main-lower-bound}. For a well-known lower bound with the standard cost model, Willard \cite{willard:loglog} defines and argues about \defn{fair} algorithms. These are algorithms where, in a fixed slot, every active packet sends with the same probability (and the probability may change from slot to slot). Fair algorithms have $\deltaMin=1$, and for a sufficiently large collision cost---specifically, $\CollCost = \Omega(n^2)$---the lower bound becomes $\tilde{\Omega}(n\sqrt{\CollCost})$. Given that \malg is fair and \whp guarantees $\tilde{O}(n\sqrt{\CollCost})$ makespan, we thus have a lower bound that is asymptotically tight to a $\poly(\log n)$-factor with our upper bound.

Our lower bound also applies in a similar way to a generalized notion of fairness, where the sending probabilities of any two packets are within some factor $\delta>0$. For example, if $\delta=\Theta(1)$, then $\deltaMin=\Theta(1)$.  Indeed, we only need such $\delta$-fairness between the two packets with the largest probabilities for our lower bound to apply, although our bound weakens as $\delta$ grows larger.

Another class of algorithms that our lower bound applies to is \defn{multiplicative weights update algorithms},  where in each slot every packet updates its sending probability by a multiplicative factor based on channel feedback in the slot. Many of these algorithms are fair (such as \cite{ChangJP19,bender:fully,richa:competitive-j,richa:efficient-j,ogierman:competitive,awerbuch:jamming,bender:coded}), but not all are (such as \cite{bender:fully}). Our lower bound applies in a non-trivial way to all such algorithms; that is, $\deltaMin>0$ given the update rules.

\section{Conclusion and Future Work}\label{sec:conclusion-future-work}

\noindent We considered a model for the problem contention resolution where each collision has cost  $\CollCost$. Our algorithm, \mainAlgorithm, addresses the static case and guarantees \whp that all packets succeed with makespan and expected collision cost that is $\tilde{O}(n\sqrt{\CollCost})$.

There are several directions for future work that are potentially interesting. First, we would like to extend this cost model to the dynamic setting (where packets may arrive over time) and design solutions. Many algorithms for the dynamic case break the packet arrivals into (mostly) disjoint batches; however, this requires periodic broadcasting that can cause collisions.

Second, for the lower bound, we believe (with some more work) that we can remove the $\poly(\log n)$ factor. However, we would like to derive a more general lower bound. A significant challenge appears to be addressing  algorithms that set up slots where exactly one packet has non-zero sending probability, while all others have zero sending probability. For example, an elected leader might ``schedule'' each packet their own exclusive slot in which to send (with probability $1$). Since there can be no collisions after such a schedule is implemented, it seems a lower bound argument must demonstrate that establishing a schedule is costly.

Third, what if collisions are not fully dictated by the actions of packets? Some wireless settings are inherently ``noisy'', due to weather conditions, co-located networks, or faulty devices. Is there a sensible model for such settings and, if so, can we reduce the cost from collisions? \smallskip

\noindent{\bf Acknowledgements.} We are grateful to the anonymous reviewers for their feedback on our manuscript. This work is supported by NSF award CCF-2144410. 

%%%%%%%%%%%%%%%%%%%%%%%%%%%%%%%%%%%%%%%%%%%%
%%%%%%%%%%%%%%%%%%%%%%%%%%%%%%%%%%%%%%%%%%%%
%%%%%%%%%%%%%%%%%%%%%%%%%%%%%%%%%%%%%%%%%%%%

%\bibliographystyle{splncs04}
%\bibliography{cr}

%%%%%%%%%%%%%%%%%%%%%%%%%%%%%%%%%%%%%%%%%%%%

%%%%%%%%%%%%%%%%%%%%%%%%%%%%%%%%%%%%%%%%%%%%

\clearpage

\appendix

\section*{Appendix}

\section{Analysis for Our Upper Bound} \label{analysis-section}

In this part of the appendix, we provide our full proofs for our upper bound, which culminate in Theorem \ref{thm:main-theorem}. We begin in Section \ref{sec:prelims} by addressing mathematical tools that are useful for our analysis. In Section \ref{sec:diagnosis-correctness}, we argue that \whp \Diagnosis correctly determines the current range and takes the correct action. We follow up in Section \ref{sec:sample-cost} with a cost analysis for the executions of \Sample. In Section \ref{sec:rundownfull}, we analyze \Rundown, showing that \whp all packets succeed, along with a bound on the expected cost. Finally, in Section \ref{sec:main-alg}, we put all of these pieces together in order to establish the correctness and bounds on the expected cost for \mainAlgorithm

For ease of presentation, our analysis assumes pessimistically that no packets succeed until they start executing \Rundown. However, our results hold if we allow packets to succeed (and then terminate) earlier in the execution \Sample, but this number is negligible (being $\tilde{O}(\sqrt{n})$).

%Note that the number of packets that would succeed in \Sample is small. The highest probability of success  occurs in the lower end of the \toolow range, when the sending probability is $1/n$. In this case, $s=\tilde{O}(\sqrt{\CollCost}) = \tilde{O}(\sqrt{n})$ packets succeed, and since there $O(\log n)$ iterations of 

\subsection{Preliminaries}\label{sec:prelims}

\noindent Throughout, we assume that $n$ is sufficiently large. We make use of the following well-known facts.

\begin{fact}\label{fact:taylor}
The following inequalities hold. 
\begin{enumerate}[label={(\alph*)},leftmargin=20pt]
\item For any $x$, $1 - x \leq e^{-x}$, 
\item For any $0\leq x<1$, $1 - x \geq e^{-x/(1-x)}$, 
\item For any $0\leq x\leq 1/2$, $1 - x \geq e^{-2 x}$. 
\end{enumerate}
\end{fact}

\noindent Note that Fact \ref{fact:taylor}(c) follows directly from Fact \ref{fact:taylor}(b); however, the former is sometimes easier to employ, and we state it explicitly to avoid any confusion. We also make use of the following standard Chernoff bounds, stated below.

\begin{theorem}\label{thm:chernoff}
(\cite{MichelGoemans}) Let $X = \sum_i X_i$, where $X_i$ are Bernoulli indicator random variables with probability $p$.  The following holds:
\begin{align*}
Pr\left(X \geq (1+\delta)E[X] \right) &\leq  \exp\left( -\frac{\delta^2 E[X]}{2+ \delta} \right) \mbox{~~~for any $\delta > 0$}\\
& \mbox{and}\\
Pr\left(X  \leq (1-\delta)E[X] \right) &\leq \exp\left( -\frac{\delta^2 E[X]}{2+\delta} \right)  \mbox{~~~for any $0 < \delta < 1$}.~\qed
\end{align*}
\end{theorem}

Throughout, we use with high probability to mean with probability  $1-O(1/n^{\kappa+1})$, where $n^\kappa$ is an upper bound on $\CollCost$ (recall Section \ref{sec:model}). 

%Consider indicator random variables $X_i$. Using Chernoff bounds, we often derive upper bounds such that with probability at least $1-p$, the value of $X_i \leq t$, where $p$ and $t$ depend on $n$ and $\CollCost$. Notably, $p$ is at least polynomially small in $n$ and $\CollCost$ (e.g., $p \leq 1/(\max\{n,\CollCost\})^2$).

%Let $X=\sum_i X_i$; that is, $X$ is the sum of random variables $X_i$. With respect to upper bounding $E[X]$, we often state guarantees in the following manner: with probability at least $1-p$, $E[X] \leq \tau$. By this, we mean that for all $i$, $Pr(X_i \leq \tau_i) \leq p_i$, where $\sum_i p_i \leq p$.  Again, the $p_i$ values and $p$ are polynomially small in $n$ and $\CollCost$. \medskip

In our first lemma below, we prove an upper bound on the probability of a collision in any slot as a function of the number of active packets. This result is a useful tool in our later arguments.

\begin{lemma}\label{lem:upper-prob-coll}
Consider any slot $t$, where there exist $m$ active packets, each sending with the same probability $p$, where  $p< 1/m$. The probability of a collision in slot $t$ is at most $\frac{2m^2p^2}{\left(1-mp\right)}$.
\end{lemma}
\begin{proof} 
Fix a time slot $t$. Let $p_i$ denote the probability that packet $i$ sends in slot $t$ and number of packets at the beginning is $m$. The probability of collision in slot $t$ is at most:
\begin{align*}
& 1- Pr(\mbox{success in slot $t$})-Pr(\mbox{empty slot in slot $t$})\\
&= 1- \binom{m}{1} p^1 (1-p)^{(m-1)}-\binom{m}{0} p^0 (1-p)^{m}\\
&\leq  1- mp(1-p)^{m} - (1-p)^m\\
&= 1- (1-p)^{m} (mp+1)\\
&\leq 1- e^{-\frac{mp}{1-p}} (mp+1) \mbox{~~~by Fact \ref{fact:taylor}(b)}\\
&\leq 1- \left(1-\frac{mp}{1-p}\right) (mp+1)\\
&\leq 1- \left(1-\frac{mp}{1-mp}\right) (mp+1) \mbox{~~~for $mp<1$}\\
&= 1- \left(mp-\frac{m^2p^2}{1-mp}+1-\frac{mp}{1-mp}\right)\\
&= \frac{m^2p^2}{1-mp}-mp+\frac{mp}{1-mp}\\
&= \frac{m^2p^2-mp+m^2p^2+mp}{1-mp}\\
&= \frac{m^2p^2-mp+m^2p^2+mp}{1-mp}\\
&= \frac{2m^2p^2}{\left(1-mp\right)}
\end{align*}
\noindent which proves the claim. \qed %\maxwell{Really only meaningful for $\CollCost \geq 2$.}
\end{proof}

\subsection{Correctness Analysis for \Diagnosis}\label{sec:diagnosis-correctness}

\noindent In order to argue correctness for \Diagnosis, we give a case analysis for each of the \tootoolow, \toolow, \goodrange, and \toohigh ranges. In each case, we provide the necessary bounds on successes and collision per sample, which allows us to demonstrate that \Diagnosis correctly diagnoses the current range and performs the correct action. We start with Lemmas~\ref{lem:generic-upper-expectation-successes}, \ref{lem:generic-lower-expectation-successes}, \ref{lem:upper-bound-successes}, and \ref{lem:lower-bound-successes}, which establish lower and upper bounds on the number of successes in the ranges of $\toolow$, $\goodrange$, and $\toohigh$. 

%Throughout this section, all of our statements hold with high probability; specifically, this means with probability at least $1-O(1/n^3)$. 

%Throughout, we let {\boldmath{$\probParam$}} $= \max\{n, \CollCost\}$. Recall From Section~\ref{sec:algorithm} that we denote the current window by $\currWin$, that $\sampleParam = \CollCost + \currWin$, and the sample size is $d\sqrt{\CollCost}\ln(\sampleParam)$.

\begin{lemma}\label{lem:generic-upper-expectation-successes}
    For any window size $\currWin=\myVar n\sqrt{\CollCost}$ where $\myVar>0$ is a constant, the expected number of successes in a sample of size $d\sqrt{\CollCost}\ln(\sampleParam)$ at most $\frac{d\ln(\sampleParam)}{\myVar}e^{-1/(2\myVar\sqrt{\CollCost})}$.
\end{lemma}
\begin{proof}
    Let $X_i$ be an independent random variable where $X_i = 1$ when $i$-th slot in $w$ contains a success, and $X_i = 0$ otherwise.
     \begin{align*}
      Pr(X_i=1) & = \binom{n}{1} \left(\frac{1}{\myVar n\sqrt{\CollCost}}\right)\left(1-\frac{1}{\myVar n\sqrt{\CollCost}}\right)^{n-1}\\
      & \leq \frac{1}{\myVar\sqrt{\CollCost}}e^{-(n-1)/(\myVar n\sqrt{\CollCost})} \mbox{~~~by Fact~\ref{fact:taylor}(a)}\\
       & \leq \frac{1}{\myVar\sqrt{\CollCost}}e^{-(n/2)/(\myVar n\sqrt{\CollCost})} \mbox{~~~since $n\geq 2$}\\
      &=\frac{1}{\myVar\sqrt{\CollCost}}e^{-1/(2\myVar\sqrt{\CollCost})}
     \end{align*}

\noindent Let $X = \sum_{k=1}^{d\sqrt{\CollCost}\ln(\sampleParam)} X_k$, then the expected number of success:
\begin{align*}
    E\left[X\right] & \leq E\left[\sum_{k=1}^{d\sqrt{\CollCost}\ln(\sampleParam)} X_k\right] \\
    & = \sum_{k=1}^{d\sqrt{\CollCost}\ln(\sampleParam)} \hspace{-12pt} E[ X_k ] \mbox{~~~by linearity of expectation}\\
    & = \frac{d\sqrt{\CollCost}\ln(\sampleParam)}{\myVar\sqrt{\CollCost}}e^{-1/(2\myVar\sqrt{\CollCost})} \\
    & = \frac{d\ln(\sampleParam)}{\myVar}e^{-1/(2\myVar\sqrt{\CollCost})}
\end{align*}
which proves the claim.\qed
\end{proof}

\begin{lemma}\label{lem:generic-lower-expectation-successes}
    For any window size $\currWin=\myVar n\sqrt{\CollCost}$ where $\myVar>0$ is a constant, the expected number of successes in a sample is at least  $\frac{d\ln(\sampleParam)}{\myVar}e^{-1/(\myVar\sqrt{\CollCost})}$.
\end{lemma}
\begin{proof}
Let $X_i$ be an independent random variable where $X_i = 1$ when $i$-th slot in $w$ contains a success, and $X_i = 0$ otherwise.
     \begin{align*}
      Pr(X_i=1) & = \binom{n}{1} \left(\frac{1}{\myVar n\sqrt{\CollCost}}\right)\left(1-\frac{1}{\myVar n\sqrt{\CollCost}}\right)^{n-1}\\
      & \geq \frac{1}{\myVar\sqrt{\CollCost}}e^{-2(n-1)/(\myVar n\sqrt{\CollCost})} \mbox{~~~by Fact~\ref{fact:taylor}(c)}\\
       & \geq \frac{1}{\myVar\sqrt{\CollCost}}e^{-1/(\myVar\sqrt{\CollCost})}.
     \end{align*}
\noindent Let $X_k$ be an independent random variable where $X_k = 1$ when $k$-th slot in the sample contains a success, and $X_k = 0$ otherwise. The expected number of successes in the sample is:

\begin{align*}
    E\left[X\right] & \geq E\left[\sum_{k=1}^{d\sqrt{\CollCost}\ln(\sampleParam)} X_k\right] \\
    & = \sum_{k=1}^{d\sqrt{\CollCost}\ln(\sampleParam)} \hspace{-12pt} E[ X_k ] \mbox{~~~by linearity of expectation}\\
    & \geq \frac{d\sqrt{\CollCost}\ln(\sampleParam)}{\myVar\sqrt{\CollCost}}e^{-1/(\myVar\sqrt{\CollCost})} \\
    & = \frac{d\ln(\sampleParam)}{\myVar}e^{-1/(\myVar\sqrt{\CollCost})}
\end{align*}
which proves the claim.\qed
\end{proof}

\begin{lemma}\label{lem:upper-bound-successes}
If $\currWin\geq an\sqrt{\CollCost}$, where $a\geq 1$ is a constant, then with probability at least $1-1/\probParam^{d'}$ the number of successes in a sample is at most  $(2d/a)\ln(\sampleParam)$, where $d'$ is an arbitrarily large constant depending on sufficiently large $d$. 
\end{lemma}
\begin{proof}
Let $\myVar=a$ in Lemma~\ref{lem:generic-upper-expectation-successes}, which implies that the expected number of successes in the sample is at most:
\begin{align*}
\frac{d\ln(\sampleParam)}{a}e^{-1/(2a\sqrt{\CollCost})} & \leq \frac{d\ln(\sampleParam)}{a}.
\end{align*}
Let $X_i=1$ if the $i$-th slot of the sample contains a success; otherwise, let  $X_i=0$. Let $X =\sum_{i=1}^{d\sqrt{\CollCost}\ln(\sampleParam)} X_i$, and note by the above that $E[X] \leq \frac{d\ln(\sampleParam)}{a}$. By Theorem~\ref{thm:chernoff}, with $\delta =1$ below, we have:

\begin{align*}
    Pr\left(X\geq (1+\delta)\frac{d\ln(\sampleParam)}{a}\right)&\leq e^{-\frac{(1)^{2}d\ln(\sampleParam)}{(1+2)a}}\\
    & = e^{-\frac{d\ln(\sampleParam)}{3a}}\\
    & = \sampleParam^{-\frac{d}{3a}}\\
    & \leq n^{-\frac{d}{3a}} \mbox{~~~since $\currWin \geq n$}\\
    & = n^{-d'}
    %& =\frac{1}{(\CollCost+\currWin)^{\frac{d}{3a}}} \mbox{~~~by definition of $\sampleParam$}\\
    %& \leq \frac{1}{(\CollCost+n\sqrt{\CollCost})^{\frac{d}{3a}}} \mbox{~~~since $w\geq n\sqrt{\CollCost}$}\\
    %& \leq \frac{1}{\probParam^{d'}} 
\end{align*}
where $d'$ can be an arbitrarily large constant, so long as $d$ is sufficiently large. \qed
\end{proof}

Our next lemma provides a useful lower bound on the number of successes.

\begin{lemma}\label{lem:lower-bound-successes}
If $n \leq w\leq bn\sqrt{\CollCost}$ for a constant $b\geq 10$, then with probability at least $1-1/\probParam^{d'}$ the number of successes exceeds $(d/2be)\ln(\sampleParam)$, where $d'$ is an arbitrarily large constant depending on sufficiently large $d$. 
\end{lemma}
\begin{proof}
Let $\myVar=b$ in Lemma~\ref{lem:generic-lower-expectation-successes}, which implies that the expected number of successes in the sample is at least:
\begin{align*}
\frac{d\ln(\sampleParam)}{b}e^{-1/(b\sqrt{\CollCost})} > \frac{d\ln(\sampleParam)}{e\,b}
\end{align*}
\noindent for any $b\geq 10$. Let $X_i=1$ if the $i$-th slot of the sample contains a success; otherwise, let  $X_i=0$. Let $X =\sum_{i=1}^{d\sqrt{\CollCost}\ln(\sampleParam)} X_i$, and note by the above that $E[X] \geq \frac{d\ln(\sampleParam)}{b e}$. By Theorem~\ref{thm:chernoff}, with  $\delta =1/2$, we have:

\begin{align*}
    Pr\left(X \leq (1-\delta) \frac{d\ln(\sampleParam)}{b e} \right)&\leq e^{-\frac{(1/2)^{2}d\ln(\sampleParam)}{(2+1/2)eb}}\\
    & = e^{-\frac{d\ln(\sampleParam)}{10eb}}\\
    & =  \sampleParam^{-\frac{d}{10eb}}\\
    %& = \frac{1}{(\CollCost + \currWin)^{\frac{d}{10eb}}}\\
     %& \leq \frac{1}{(\CollCost + n)^{\frac{d}{10eb}}} \mbox{~~~since $\currWin \geq n$}\\
    & \leq  \probParam^{-\frac{d}{10eb}}\\
    & = \frac{1}{\probParam^{d'}}
\end{align*}
\noindent where $d'$ can be an arbitrarily large constant, so long as $d$ is sufficiently large. \qed
\end{proof}

\noindent The next two lemmas pertain to the range of \tootoolow. Lemma~\ref{l:all-collisions-if-smallW} shows that, when the window size is $O(n/\log n)$, the entire sample will consist of collisions. Unlike many of our other arguments, we cannot employ a Chernoff to establish this fact. Lemma~\ref{l:upper-too-too-low} handles the remainder of the \tootoolow range, showing that much of the sample will consist of collisions. 

\begin{lemma}\label{l:all-collisions-if-smallW}
    Suppose that $1\leq \currWin \leq \frac{n}{c\ln n}$ and $n$ is sufficiently large and $c\geq 1$ is a constant. Then, with probability at least $1-1/n^{d'}$, the entire sample consists of collisions, where $d'$ depends on sufficiently large $c$.
\end{lemma}
\begin{proof}
    For any fixed slot, the probability of a collision is at least $1$ minus the probability of a success, minus the probability of that the slot is empty. In other words, the probability of a collision is at least:
    \begin{align*}
      & \mbox{~~~~~}1 - \left(1-\frac{c\ln n}{n}\right)^{n} - \binom{n}{1}\left(\frac{c\ln n}{n}\right)\left(1 - \frac{c\ln n}{n}\right)^{n-1}\\
      &  \geq 1 - \frac{1}{e^{c\ln n}} - \left( \frac{c\ln n}{e^{c(n-1)\ln(n)/n}}\right) \mbox{~~~by Fact~\ref{fact:taylor}(a)}\\
      & \geq 1 - \frac{1}{e^{c\ln n}} - \left( \frac{c\ln n}{e^{c\ln(n)/2}}\right) \mbox{~~~since $\frac{n-1}{n} \geq \frac{1}{2}$ for $n\geq 2$}\\
      & = 1 - \frac{1}{n^c} - \frac{c\ln n}{n^{c/2}}\\
      & \geq 1 - \frac{1}{n^{d'}} \mbox{~~~for $n$ sufficiently large}
    \end{align*}
\noindent which yields the claim. \qed  
\end{proof}

\begin{lemma}\label{l:upper-too-too-low}
Suppose $\frac{n}{c\ln n} < \currWin \leq n$, where $c\geq 1$ is any constant. Then, with probability at least $1-1/n^{d'}$, the  sample contains at least $d\sqrt{\CollCost}\ln(\sampleParam)/(8e^2)$ slots consists of collisions, where $d'$ is an arbitrarily large constant depending on sufficiently large $d$. 
\end{lemma}
\begin{proof}
Let the indicator random variable $X_i=1$ if the $i$-th slot of the $d\sqrt{\CollCost}\ln(\sampleParam)$ slots contains a collision; otherwise, let  $X_i=0$. The probability of a collision is lower bounded as follows.
   \begin{align*}
      Pr(X_i=1) & \geq \binom{n}{2} \left(\frac{1}{\currWin}\right)^2\left(1-\frac{1}{\currWin}\right)^{n-2}\\
      & \geq \binom{n}{2} \left(\frac{1}{\currWin}\right)^2\left(1-\frac{1}{ \currWin}\right)^{n}\\
      &\geq \left(\frac{n(n-1)}{2\currWin^2}\right) e^{-2n/\currWin} \mbox{~~~by Fact~\ref{fact:taylor}(c)}\\
      &\geq \frac{(n-1)}{2ne^{2}} \mbox{~~~minimized when $\currWin=n$}\\
      & \geq \frac{1}{4e^2} \mbox{~~~since $\frac{n-1}{n} \geq \frac{1}{2}$ for $n\geq 2$.}
     \end{align*} 
%\noindent Since the probability of a collision is minimized when $\currWin=n$, we 

Let  $X =\sum_{i=1}^{d\sqrt{\CollCost}\ln (\currWin)} X_i$. The expected number of collisions in the sample of ${d\sqrt{\CollCost}\ln(\currWin)}$ slots is:\\
\begin{align*}
    E\left[X\right] & \geq E\left[\sum_{i=1}^{d\sqrt{\CollCost}\ln(\currWin)} \hspace{-8pt}X_i\right] \\
    & = \frac{d\sqrt{\CollCost}\ln(\currWin)} {4e^{2}} \mbox{~~~by linearity of expectation.}
\end{align*}
We then employ Theorem~\ref{thm:chernoff}, with $\delta =1/2$, to obtain:
\begin{align*}
    Pr\left(X \leq (1-\delta) \frac{d \sqrt{\CollCost}\ln(\currWin)}{4e^{2}} \right)&\leq e^{- \frac{{(1/2)}^2d\sqrt{\CollCost}\ln(\currWin)}{(5/2)4e^{2}}}\\
    & =e^{-\frac{d\sqrt{\CollCost}\ln(\currWin)}{40e^{2}}}\\
    & =\currWin^{-\frac{d\sqrt{\CollCost}}{40e^2}}\\
    & \leq (n/c\log n)^{-\frac{d\sqrt{\CollCost}}{40e^2}}\\
    & \leq n^{-d'}
\end{align*}
\noindent where $d'$ can be an arbitrarily large constant, so long as $d$ is sufficiently large. \qed
\end{proof}

\begin{lemma}\label{l:same-view}
Under \Diagnosis, all packets witness the same number of empty slots, successes, and collisions.
\end{lemma}
\begin{proof}
In each slot, a packet is either sending or listening. First, consider the case where at least one packet sends. Each such sending packet knows whether it succeeded (i.e., a success) or it failed (i.e., a collision). Each non-sending packet listens and hears the success (i.e., a success) or the failure (i.e., a collision). Second, if no packet sends, then every packet is listening and hears the same thing: silence. Therefore, in all cases, each packet witnesses the same outcome of the slot. \qed
\end{proof}

Lemma~\ref{l:same-view} guarantees that all packets are always in the same window (or have the same sending probability) when executing \Diagnosis, and they always take the same action in an execution of \Diagnosis. The following lemmas establish that \Diagnosis takes the {\it correct} action in each execution.

\begin{lemma}\label{l:okay-TTL}
With probability at least $1-O(1/n^{d'})$, if $\currWin \in \tootoolow$, then all packets double their window, where $d'$ is an arbitrarily large constant depending on sufficiently large $d$. 
\end{lemma}
\begin{proof}
    We do a case analysis. For the first case, consider that $\currWin \in [1, n/(c\ln n))$ for some sufficient large constant $c\geq 1$. By Lemma~\ref{l:all-collisions-if-smallW}, with probability at least $1-n^{-d'}$, all slots in the sample contain a collision. Therefore, \whp, the if-statement on Line \ref{alg:first-if} is not entered and, instead, the  matching else-statement on Line \ref{alg:first-else} is entered. Since all sample slots are collisions, the  if-statement on Line \ref{alg:third-else} is entered, and the window doubles.

    For the second case, consider $\currWin \in [n/(c\ln n), n)$. Either we (a) enter the first if-statement on Line \ref{alg:first-if}, or we (b) enter the matching else-statement on Line \ref{alg:first-else}.
    \begin{itemize}
    \item {\bf Subcase (a).} If the second if-statement on Line \ref{alg:second-if} is entered, then we note by Lemma \ref{l:upper-too-too-low}, \whp, that we have at least $d\sqrt{\CollCost}\ln(\currWin)/(8e^2)$  collisions, which means we double the window, as desired. Otherwise, we execute Line \ref{alg:second-else}, and the window doubles.

    \item {\bf Subcase (b).} By Lemma \ref{l:upper-too-too-low} \whp we have at least $d\sqrt{\CollCost}\ln(\currWin)/(8e^2)$  collisions, which means we enter the if-statement on Line \ref{alg:third-else} and double the window. 
    \end{itemize}
    \noindent In both cases, the packet doubles its window. \qed
\end{proof}

\begin{lemma}\label{l:okay-TL}
With probability $1-O(1/n^{d'})$, if $\currWin \in \toolow$, then all packets double their window,  where $d'$ is an arbitrarily large constant depending on sufficiently large $d$. 
\end{lemma}
\begin{proof}
By Lemma~\ref{lem:lower-bound-successes} with $b=10$, for $w$ such that $n \leq w < 10 n\sqrt{\CollCost}$, with probability at least $1-n^{-d'}$ the number of successes exceeds $\frac{d}{20e}\ln(\currWin)$. Therefore,  the first if-statement on Line \ref{alg:first-if} is entered, but the next if-statement on Line \ref{alg:second-if} is not. Thus, Line \ref{alg:second-else} will be executed and so the window will double, as claimed. \qed
\end{proof}

The following lemma places an upper bound on the number of collisions when the window is at least $200n\sqrt{\CollCost}$. This allows us to argue that when we are in the \goodrange and \toohigh ranges, we avoid executing the lines in \Diagnosis that would double the window. 

We note that the next lemma gives a very loose upper bound on the number of collisions. In fact, if we expected even a single collision in this range, our claimed upper bound would not hold. The careful reader will note that (in our proof) the expected number of collisions is $\tilde{O}(1/\sqrt{\CollCost})$; however, a loose upper bound that holds with high probability is sufficient to construct the rules for \Diagnosis.

\begin{lemma}\label{l:sizable-window-upper-collisions}
With probability at least $1-1/n^{d'}$, if $\currWin\geq 200n\sqrt{\CollCost}$, then the number of collisions in the sample is at most $(d/90)\ln(\currWin)$, where $d'$ is an arbitrarily large constant depending on sufficiently large $d$.
\end{lemma}
\begin{proof}
By Lemma~\ref{lem:upper-prob-coll},  the probability of a collision in a slot is at most $\frac{2n^2p^2}{\left(1-np\right)}$, where $p\leq 1/(200n\sqrt{\CollCost})$. Plugging in, we obtain that the probability of collision is upper bounded by:

\begin{align*}
    \frac{2n^2p^2}{\left(1-np\right)} & \leq \frac{2n^2(1/(200n\sqrt{\CollCost}))^2}{\left(1-1/(200\sqrt{\CollCost})\right)}\\
    & = \frac{1}{200 \CollCost (1-1/200)}\\
    & \leq \frac{1}{100\CollCost}
\end{align*}
Let the indicator variable $X_i=1$ if the $i$-th slot is a collision; otherwise, let $X_i=0$.
\noindent Let $X =\sum_{i=1}^{d\sqrt{\CollCost}\ln(\currWin)} X_i$. 
The expected number of collisions in the sample of ${d\sqrt{\CollCost}\ln{\currWin}}$ slots is:\\
\begin{align*}
    E\left[X\right] & \leq E\left[\sum_{i=1}^{d\sqrt{\CollCost}\ln(\currWin)} \hspace{-6pt}X_i\right] \\
    & = \sum_{i=1}^{d\sqrt{\CollCost}\ln(\currWin)} \hspace{-6pt}E[X_i] \mbox{~~~by linearity of expectation} \\
    & \leq \frac{d\ln(\currWin)} {100\sqrt{\CollCost}}\\
    & \leq \frac{d\ln(\currWin)}{100}
\end{align*}
Letting $\delta =1/10$ in Theorem~\ref{thm:chernoff}, we have:\\
\begin{align*}
    Pr\left(X \geq (1+1/10) \frac{d\ln (\currWin)}{100} \right) & \leq Pr\left(X \geq  \frac{d\ln (\currWin)}{90} \right)\\
    &\leq e^{- \frac{(1/10)^2 d\ln(\currWin)}{2+(1/10)}}\\
    & \leq e^{-\frac{d\ln(\currWin)}{210}}\\
     & \leq \frac{1}{\currWin^{d'}}\\
    & \leq \frac{1}{n^{d'}} \mbox{~~~since $\currWin \geq n$}
\end{align*}
\noindent where $d'$ can be arbitrarily large depending only on $d$. \qed
\end{proof}

%\noindent\maxwell{I'm adding a gap above the good range too, it now goes from $200n\sqrt{\CollCost}$ to $10^3 n\sqrt{\CollCost}$. We have to say what happens in this gap, and the other gap (between too low and good range).}\\

%\noindent\umesh{Why this gap is required? Are we using this count inside \Diagnosis?} \maxwell{We don't need to change the algorithm. But we need a gap, since otherwise we cannot show that for Lemma \ref{l:okay-GR} the number of successes is sufficiently large.}

\noindent We now argue that \Diagnosis takes the correct action in each of the \goodrange and \toohigh ranges.

%\maxwell{modify this to call the range \goodrange, and we need to update the figure and possibly the algorithm description, depending on what is written there for these ranges.}

\begin{lemma}\label{l:okay-GR}
With probability at least $1-O(1/n^{d'})$, if $\currWin \in \goodrange$, then all packets execute \Rundown,  where $d'$ is an arbitrarily large constant depending on sufficiently large $d$.
\end{lemma}
\begin{proof}
Since $\currWin\in\goodrange$, this means that $200n\sqrt{\CollCost}  \leq \currWin \leq 10^3n\sqrt{\CollCost}$. By Lemma \ref{lem:lower-bound-successes} with $b=10^3$, the number of successes in the sample exceeds $\frac{d\ln{\currWin}}{2e\,10^3} > \frac{2d\ln(\currWin)}{10^5}$. Therefore, the if-statement on Line \ref{alg:first-if} is entered. By Lemma~\ref{lem:upper-bound-successes} with  $a=200$, the number of  successes is at most  $\frac{2d\ln(\currWin)}{200} = \frac{d\ln(\currWin)}{100} \leq \frac{d\ln(\currWin)}{20e}$; thus, the if-statement on Line \ref{alg:second-if} is entered. By Lemma~\ref{l:sizable-window-upper-collisions},  the number of collisions is at most $(d/90)\ln(\currWin) < d\sqrt{\CollCost}\ln(\currWin)/(8e^2)$, so Line \ref{alg:if-if-if} is not entered, and instead Line \ref{alg:execute-rundownfull} is entered, which executes \Rundown.  \qed
\end{proof}

\begin{lemma}\label{l:okay-TH}
With probability at least $1-O(1/n^{d'})$, if $\currWin \in \toohigh$, then all packets halve their window, where $d'$ is an arbitrarily large constant depending on sufficiently large $d$.
\end{lemma}
\begin{proof}
If $\currWin \in \toohigh$, then $\currWin > 10^5 n\sqrt{\CollCost}$. By  Lemma~\ref{lem:upper-bound-successes} with $a=10^5$, the number of successes is at most $(2d/10^5)\ln (\currWin)$. Thus, the else-statement on Line \ref{alg:first-else} is entered. By Lemma \ref{l:sizable-window-upper-collisions}, the number of collisions is at most $(d/90)\ln(\currWin) < \frac{d\sqrt{\CollCost}\ln(\currWin)}{8e^2}$, and so Line \ref{alg:halve-window} is entered and the window is halved. The total error associated with the lemmas invoked is $O(1/n^{d'})$. \qed
\end{proof}

\noindent Finally, what about the ``uncertain'' intervals $[10 n\sqrt{\CollCost}, 200 n\sqrt{\CollCost})$ and $(10^3 n\sqrt{\CollCost}, 10^5 n\sqrt{\CollCost}]$? Our next two lemmas address this aspect by showing that the \Diagnosis revises the current window towards the \goodrange or executes \Rundown  directly.

\begin{lemma}\label{l:limbo-low}
With probability at least $1-O(1/n^{d'})$, if $\currWin \in \uncertainA$, then either the window doubles or \Rundown is executed, where $d'$ is an arbitrarily large constant depending on sufficiently large $d$.
\end{lemma}
\begin{proof}
Since $\currWin \in \uncertainA$, we have that $\currWin \in [10 n\sqrt{\CollCost}, 200 n\sqrt{\CollCost})$. By Lemma \ref{lem:lower-bound-successes} with $b=200$ (the part of \uncertainA where successes are most scarce), the number of successes in the sample still exceeds $\frac{d \ln{\currWin}}{400e}$, which exceeds $\frac{2d\ln(\currWin)}{10^5}$, so Line \ref{alg:first-if} will be true.   For the remaining lines that may be executed---Lines \ref{alg:second-if} to \ref{alg:double-blah}--- we either double the window (due to many collisions) or execute \Rundown, as claimed. \qed
\end{proof}

%\maxwell{Things we need to check:}

%In \tootoolow, do we have more than $(d/90)\ln(\currWin)$
%In \toolow, do we have more than $(d/90)\ln(\currWin)$?
%In \uncertainA, do we have more than $(d/90)\ln(\currWin)$?

\begin{lemma}\label{l:limbo-high}
With probability at least $1-O(1/n^{d'})$, if $\currWin \in \uncertainB$, then either the window halves or \Rundown is executed, where $d'$ is an arbitrarily large constant depending on sufficiently large $d$.
\end{lemma}
\begin{proof}
We prove this by ruling out all the ways in which $\currWin$ could (incorrectly) double. 

Unlike our argument for $\currWin\in\goodrange$, we cannot guarantee that Line  \ref{alg:first-if} will be executed, since our window might be large, say just shy of $10^5 n\sqrt{\CollCost}$, (recall that  \uncertainB $=[10^3 n\sqrt{\CollCost}, 10^5 n\sqrt{\CollCost})$), and we cannot prove a sufficient number of successes. However, if we proceed to Line \ref{alg:first-else}, then by Lemma \ref{l:sizable-window-upper-collisions}, the number of collisions is at most $(d/90)\ln(\currWin)$, which means there are not enough collisions to meet the if-statement condition on \ref{alg:third-else}, and thus the window will not be doubled here.

Another way in which the window can (incorrectly) double is via Line \ref{alg:second-else}. However, we only execute Line \ref{alg:second-else} if Line \ref{alg:second-if} fails to hold; that is, in the case that the sample contains {\it more} than than $\frac{d\ln(\currWin)}{20e}$ successes. However, by Lemma~\ref{lem:upper-bound-successes} with $a=10^3$ (the part of \uncertainB where successes are most likely), the sample contains at most $\frac{2d\ln(\currWin)}{10^3}  < \frac{d\ln(\currWin)}{20e}$ successes, and so  Line \ref{alg:second-if} will indeed execute.

Finally, in the case that Line \ref{alg:second-if} executes, we do not have enough collisions to satisfy Line \ref{alg:third-else}, again by  Lemma \ref{l:sizable-window-upper-collisions}. Thus, the window cannot double via Line \ref{alg:double-the-first}

We have shown that, when $\currWin \in \uncertainB$, either the window must halve, or  \Rundown is executed, as claimed. \qed
\end{proof}

\begin{comment}
\paragraph{With High Probability in {\boldmath{$n$}} and {\boldmath{$\CollCost$}}.} To wrap up this section, we note that \mainAlgorithm executes \Diagnosis $k\ln(\CollCost)$ times, where $k>0$ is a large constant. As described in Section~\ref{sec:algorithm}, this results in $k\ln(\CollCost)$ votes, and then \mainAlgorithm takes action corresponding to whichever option---halve $\currWin$, double $\currWin$, or execute \DiagnosisC---has the most votes. 

Each of the corresponding $30\ln(\CollCost)$ samples generated by \Sample is independent,  the corresponding vote output by \Diagnosis is correct with probability at least $1-O(1/n^{3})$. Define the indicator random variable $Y_i=1$ if the $i$-th sample results in a  correct vote by  \Diagnosis; otherwise, $Y_i=0$.  Let $Y=\sum_{i=1}^{30\ln(\CollCost)} Y_i$ count the number of correct votes, and note that $E[Y] \geq  (1-1/n^3)30\ln(\CollCost)$. By Theorem \ref{thm:chernoff} with $\delta = 1/4$, we have:
\begin{align*}
    Pr\left( Y \leq  (1-\delta) (1-1/n^3)30\ln(\CollCost)\right) & =  Pr\left( Y \leq  (3/4)(1-1/n^3)30\ln(\CollCost)\right) \\
    &\leq e^{-(3/4)^2 (1-1/n^3)30\ln(\CollCost) /(11/4)}\\
    & = e^{ -(9/44)(1-1/n^3)30\ln(\CollCost)}\\
    & = e^{ -(9/88) 30\ln(\CollCost)}\\
    & = e^{ -3\ln(\CollCost)} \\
    & = \frac{1}{\CollCost^{3}}
\end{align*}

Therefore, while all of our claims in this section hold with probability at least $1-O(1/n^3)$, they also hold with probability at least $1-O(1/\CollCost^3)$, and we make use of whichever guarantee is strongest.
\end{comment}

\subsection{Cost Analysis for \Sample} \label{sec:sample-cost}

\noindent The next lemmas establish bounds on the expected collision cost for executing \Sample. To this end, we employ the results established above, which hold with probability at least $1 - 1/n^{d'}$ for as large a constant {\boldmath{$d$}} $'\geq 1$ as we wish, depending only on sufficiently large $d$. Thus, the results in this section also hold with this high probability, and for ease of presentation, we omit this from our technical lemma statements and arguments.

\begin{lemma}\label{lm-tl-exp-collisions}
If $\currWin \in$ \tootoolow $\cup$ \toolow, the execution of \Sample has  $O(n\sqrt{\CollCost}\ln(n))$ expected collision cost.
\end{lemma}
\begin{proof} We prove this using a case analysis. \medskip

\noindent{\bf Case 1: {\boldmath{$\CollCost \leq 2n$}}.} This captures all of \tootoolow and the left endpoint of \toolow. Here, even if the entire sample consists of collisions, then the cost is at most:
\begin{align*}
(d\sqrt{\CollCost}\ln(\currWin)) \CollCost & \leq (d \sqrt{\CollCost}\ln(n\sqrt{\CollCost})) 2n\\
&= O(n\sqrt{\CollCost}\ln(n))
\end{align*}

\noindent{\bf Case 2: {\boldmath{$\CollCost > 2n$}}.} With high probability,  $\currWin \in$ \toolow only if the initial window lies in \toolow. (With high probability, we cannot have started in \tootoolow and moved up, since we start at a window of size $\CollCost > n$ in this case.)

The probability of a collision is maximized when the window is smallest; that is, when the window has size $\CollCost$ (which is the initial window size), since \whp the window only increases (given that we are in \toolow). By Lemma \ref{lem:upper-prob-coll}, the probability of a collision is at most:
\begin{align*}
    \frac{2n^2p^2}{1-np} &\leq   \frac{2n^2 (1/\CollCost)^2}{1-n(1/\CollCost)}\\
    & \leq 4n^2/\CollCost^2  \mbox{~~~since $\CollCost > 2n$}
\end{align*}
\noindent where the equation from Lemma \ref{lem:upper-prob-coll} is applicable, since $p=1/\CollCost < 1/(2n)$. Thus, expected cost is at most:
\begin{align*}
   \left( d\sqrt{\CollCost}\ln(\currWin)\right) \left(\frac{4n^2}{\CollCost^2}\right) \CollCost &= \frac{4dn^2\ln(\currWin)}{\sqrt{\CollCost}}\\
   &< \frac{4 d n\CollCost\ln(\currWin)}{\sqrt{\CollCost}} \mbox{~~~since $\CollCost > 2n$}\\
   &= 4d n\sqrt{\CollCost}\ln(\currWin)\\
   &= 4d n\sqrt{\CollCost}\ln(10n\sqrt{\CollCost}) \mbox{~~~since $\currWin  \in$ \toolow}\\
   &=O(n\sqrt{\CollCost}\ln(n)) \mbox{~~~since $\mathcal{\CollCost} = O(\poly(n))$}
\end{align*}
\noindent as claimed. \qed
\end{proof}

\begin{lemma}\label{lm-goodrange-exp-collisions}
If $\currWin \in \uncertainA \cup \goodrange \cup \uncertainB$, then the execution of \Sample has $O(\sqrt{\CollCost}\ln(n))$ expected collision cost.
\end{lemma}
\begin{proof}
By Lemma~\ref{lem:upper-prob-coll} with $m = n$ and $p\leq \frac{1}{n\sqrt{\CollCost}}$, the probability of a collision is at most:
 \begin{align*}
      \frac{2m^2p^2}{1-mp} &\leq  \frac{2n^2\frac{1}{n^2\CollCost}}{1-\frac{n}{n\sqrt{\CollCost}}} \mbox{~~~since this value is largest for $p=\frac{1}{n\sqrt{\CollCost}}$}\\
      & =\frac{2}{\CollCost\left(1-\frac{1}{\sqrt{\CollCost}}\right)}\\
      &= O(1/\CollCost) 
    \end{align*}

\noindent Let $X=\sum_{i=1}^{d\,\sqrt{\CollCost}\ln(\currWin)} X_i$. The expected number of collisions in the sample of ${d\,\sqrt{\CollCost}\ln(\currWin)}$ slots is:\\
\begin{align*}
    E\left[X\right] & \geq E\left[\sum_{i=1}^{d\,\sqrt{\CollCost}\ln(\currWin)} X_i\right]\\ 
      & = \sum_{i=1}^{d\sqrt{\CollCost}\ln(\currWin)} \hspace{-6pt}E[X_i] \mbox{~~~by linearity of expectation} \\
    &= O\left(\frac{d\,\sqrt{\CollCost}\ln(\currWin)} {\CollCost}\right)\\ 
    & = O\left(\frac{d\ln(n)}{\sqrt{\CollCost}}\right) \mbox{~~~since $\CollCost=\poly(n)$}
\end{align*}
\noindent Thus, the expected collision cost is at most:
\begin{align*}
    O\left(\frac{d\ln(n)}{\sqrt{\CollCost}}\right) \CollCost = O(\sqrt{\CollCost} \ln(n) )
\end{align*}
\noindent as claimed.\qed
\end{proof}

\begin{lemma}\label{lm-th-exp-collisions}
If $\currWin \in$ \toohigh, the execution of \Sample results in an $O(\sqrt{\CollCost}\ln(n))$ expected collision cost.
\end{lemma}
\begin{proof}
For \toohigh, by using Lemma ~\ref{lem:upper-prob-coll}, we plug in the value $m = n$ (number of packets) and $p\leq \frac{1}{10^5n\sqrt{\CollCost}}$ (sending probability). The probability of a collision is at most:
 \begin{align*}
      \frac{2m^2p^2}{1-mp} & \leq \frac{2n^2\frac{1}{n^2 10^{10}\CollCost}}{1-\frac{n}{10^5 n\sqrt{\CollCost}}}  \mbox{~~~since this value is largest for $p=\frac{1}{10^5n\sqrt{\CollCost}}$}\\
      & =\frac{2}{10^{10}\CollCost\left(1-\frac{1}{10^{5}\sqrt{\CollCost}}\right)}\\
      & = \Theta\left( \frac{1}{\CollCost}\right)
      %& < 1 ~~~~~~\mbox{for $\CollCost \geq 2$}
    \end{align*}

We have a sample of size $d\,\sqrt{\CollCost}\ln(\currWin)$. \noindent Let $X=\sum_{i=1}^{d\,\sqrt{\CollCost}\ln(\currWin)} X_i$. The expected number of collisions in the sample of ${d\sqrt{\CollCost}\ln(\currWin)}$ slots is:\\
\begin{align*}
    E\left[X\right] & = E\left[\sum_{i=1}^{d\sqrt{\CollCost}\ln (\currWin)} X_i\right] \\
       & = \sum_{i=1}^{d\sqrt{\CollCost}\ln(\currWin)} \hspace{-6pt}E[X_i] \mbox{~~~by linearity of expectation} \\
    & = O\left(\frac{d\sqrt{\CollCost}\ln(\currWin)} {C}\right)\\
    & =O\left(\frac{d\ln(n)}{\sqrt{\CollCost}}\right) \mbox{~~~since $\CollCost=\poly(n)$.}
\end{align*}
Thus, multiplying by $\CollCost$, we obtain that the expected collision cost is at most $O(\sqrt{\CollCost} \ln(n))$. \qed
\end{proof}

\begin{lemma}\label{lem:correctness-cost-collectsample-diagnosis}
The executions of \Sample and \Diagnosis guarantee that: (i) \Rundown is executed within  $O(\sqrt{\CollCost}\log^2(n) )$ slots,  and (ii) the total expected collision cost until that point is $O\left(n\sqrt{\CollCost} \log^2(n)\right)$.
\end{lemma}
\begin{proof}
We first bound the number of slots, and then we bound the expected collision cost.\medskip

\noindent{\it (i) Bound on Number of Slots.} By Lemmas \ref{l:okay-TTL}, \ref{l:okay-TL}, and \ref{l:limbo-low}, \Diagnosis correctly doubles the window or executes \Rundown. Similarly, by Lemmas \ref{l:okay-TH} and \ref{l:limbo-high},  \Diagnosis correctly halves the window or executes \Rundown. Therefore, \Diagnosis is always making progress towards reaching  \goodrange.   

How many window halvings or doublings are required to reach  \goodrange? We start at an initial window of size $\CollCost$, and either \goodrange is smaller---in which case, we must perform  $O(\log(\CollCost))=O(\log(n))$ samples---or larger---in which case, we must again perform $O(\log (n\sqrt{\CollCost})) = O(\log(n))$ samples. Thus, the number of slots until \goodrange is reached is at most the sample size multiplied by this number of samples, which is $O(\sqrt{\CollCost}\log^2(n))$. Then, by Lemma~\ref{l:okay-GR}, \Rundown is executed.\medskip

\noindent{\it (ii) Bound on Expected Collision Cost.} Let $S_i$ denote the collision cost for sample $i$ when executing \Diagnosis, where $i=1, ..., h$, where $h=O(\log(n))$ is the number of samples. Denote the  total cost from collisions by $S=\sum_{i=1}^h S_i$. By 
Lemmas \ref{lm-tl-exp-collisions}, \ref{lm-goodrange-exp-collisions}, and \ref{lm-th-exp-collisions},  we have:
\begin{align*}
    E[S] &= O\left(h\,n\sqrt{\CollCost} \log(n) \right)\\
    & = O\left(n\sqrt{\CollCost} \log^2(n)\right)
\end{align*}
as claimed. \qed
\end{proof}

%%%%%%%%%%%%%%%%%%%%%%%%%%%%%%%%%%%%%%%%%%%%%%%%%%
%%%%%%%%%%%%%%%%%%%%%%%%%%%%%%%%%%%%%%%%%%%%%%%%%%
%%%%%%%%%%%%%%%%%%%%%%%%%%%%%%%%%%%%%%%%%%%%%%%%%%

\subsection{Correctness and Cost Analysis for \Rundown}\label{sec:rundownfull}

\noindent A single execution of lines \ref{alg:run-start} to \ref{alg:run-end} in \Rundown of \mainAlgorithm is called a \defn{run}. In any run, all packets operate over a window of size $\currWin$ and send themselves in each slot with probability $2/\currWin$. At the end of the window, those packets that succeeded will terminate, while the remaining packets will carry on into the next window, which has size $\currWin/2$.  

Our overall goal is to show that (roughly) a logarithmic in $n$ runs are sufficient for all packets to succeed. To this end, our next lemma (below) argues that in each run at least half of the remaining packets succeed. Again, throughout this section, our results hold with a tunable probability of error that is polynomially small in $n$, and we omit this in our lemma statements.

\begin{lemma}\label{l:single-rundown}
Suppose $m$ packets execute a window of \Rundown\hspace{-2pt}($\currWin$), where $\gamma\ln n\leq m\leq n$, $\gamma\geq 1$ is a positive constant,  and $\currWin\geq 8m\sqrt{\CollCost}$. Then,  at least $m/2$ packets succeed.
\end{lemma}
\begin{proof}
Fix a packet and calculate the probability of failing over the entire window of size $8m\sqrt{\CollCost}$. To do so, note that the probability that this packet succeeds in a fixed slot is: (i) the probability that the packet sends in this slot, multiplied by (ii) the probability that all other $m-1$ packets remain silent. That is:
$$\left( \frac{2}{\currWin}\right)\left(1-\frac{2}{\currWin}\right)^{m-1}.$$

\noindent Thus, the probability of failure in the slot is:

$$1-\left( \frac{2}{\currWin}\right)\left(1-\frac{2}{\currWin}\right)^{m-1}.$$

\noindent It follows that the probability of failing over all $\currWin$ slots is:
\begin{align}
\left(1-\frac{2}{\currWin}\left(1-\frac{2}{\currWin}\right)^{m-1}\right)^{\currWin} & \leq \left(1-\frac{2}{\currWin}\frac{1}{e^{\frac{4(m-1)}{\currWin}}}\right)^{\currWin} \mbox{~~~by Fact \ref{fact:taylor}(c)} \nonumber\\
& \leq \left(1-\frac{2}{\currWin}\frac{1}{e^{\frac{4m}{\currWin}}}\right)^{\currWin} \nonumber\\
& \leq e^{-\frac{2\currWin}{(\currWin) (e^{4m/\currWin})}}\mbox{~~~by Fact \ref{fact:taylor}(a)} \nonumber\\
& \leq e^{-\frac{2}{e^{4m/\currWin}}} \nonumber\\
& \leq \frac{1}{e^{\frac{2}{e^{4m/\currWin}}}} \nonumber\\
& \leq \frac{1}{e^{\frac{2}{e^{4m/(8m\sqrt{\CollCost})}}}} \mbox{~~~since this is maximized when $\currWin = 8m\sqrt{\CollCost}$} \nonumber\\
& \leq \frac{1}{e^{\left(\frac{2}{ e^{1/(2\sqrt{\CollCost})}  }\right)}} \nonumber\\
%& = \frac{1}{e^{\frac{2}{e^{\frac{1}{2}\frac{1}{\sqrt{\CollCost}}}}}} \nonumber\\
%& = \frac{1}{e^{\frac{2}{\sqrt{e}^{\frac{1}{\sqrt{\CollCost}}}}}} \nonumber \label{eqn:prob-fail}\\
& \leq 0.3 \mbox{~for all~} \CollCost\geq 1. \label{eqn:prob-fail}
\end{align}

\noindent For $i=1, ..., m$, let the indicator random variable $X_i=1$ if packet $i$ fails over the entire window; otherwise, let $X_i=0$. We let $X=\sum_{i=1}^m X_i$, and note that:
\begin{align*}
   E[X] & \leq E\left[\sum_{i=1}^{m} X_i\right]\\
    & \sum_{i=1}^m E[X_i] \mbox{~~~by linearity of expectation}\\
    & =  \sum_{i=1}^m Pr(X_i=1) \\
    & \leq 0.3 m
\end{align*}

\noindent By Theorem \ref{thm:chernoff}, letting $\delta=2/3$, we have:
\begin{align*}
    Pr\left( X \geq (1+\delta) 0.3m \right) &=  Pr\left( X \geq  m/2 \right)  \\
    & \leq e^{-(2/3)^2 0.3m/(8/3)}\\
    & = e^{-m/20}\\
    & \leq e^{-(\gamma/20)\ln n} \mbox{~~~since $m\geq \gamma\ln n$}\\
    & = 1/n^{d'}.
\end{align*}
Therefore, with probability at least $1-1/n^{d'}$, at least half of the packets succeed. \qed
\end{proof}

\noindent{\bf Discussion of Lemma~\ref{l:single-rundown}.} We highlight that Lemma \ref{l:single-rundown} applies to all windows in a fixed run. That is, in the first window of the run we have $m=n$ and $\currWin = kn\sqrt{\CollCost}$ for some constant $k\geq8$. In the second window of the run, we have $m\leq n/2$ and $\currWin = k(n/2)\sqrt{\CollCost}$. Referring to Line \ref{alg:while} of \mainAlgorithm, we see that this process---whereby the number of packets is {\it at least} halved, and the window size is halved---stops once we reach a window in the run with size less than $\Theta(\lg(w_0)\sqrt{\CollCost})$, where $w_0 = \Theta(n\sqrt{\CollCost})$. 

How many windows are in the run? That is, how many times can we halve our initial window of size $w_0=kn\sqrt{\CollCost}$ before it drops below size $\lg(w_0)\sqrt{\CollCost} = \lg(kn\sqrt{\CollCost})\sqrt{\CollCost}$? We solve for the number of runs $j$ in:
$$\frac{kn\sqrt{\CollCost}}{2^j} = \lg\left(kn\sqrt{\CollCost}\right)\sqrt{\CollCost}$$
\noindent which yields $j = \lg(kn) - \lg\lg(kn\sqrt{\CollCost})$. Since $\CollCost=\poly(n)$, this means $j = \lg(n) - \lg\lg(n) + O(1)$. This implies that, over these $j$ windows in the run, $n$ packets will be reduced to no more than $n/2^j = O(\log n)$ packets by the final window. Therefore, we have the following corollary.

\begin{corollary}\label{cor:rundown-almost-all-packets-succeed}
The number of windows in the run in \Rundown is $\lg(n) - \lg\lg(n) + O(1)$, after which only $O(\log(n))$ packets remain.
\end{corollary}

% The smallest number window halvings is at least  8n\sqrt{C} / 2^i = 60\ln(C)

% lg( (2/15)n C^{1/2} /ln(C)) = lg(2/15) + lg(n) + (1/2)lg(C) - lg(ln(C))
% =  -3

For the remainder of this section, all of the lemmas that we invoke hold with high probability in $n$, and thus we omit this from most of our lemmas statements and arguments. We are now able to prove correctness and an upper bound on the number of slots incurred by \Rundown.    

\begin{lemma}\label{lem:correctness-latency-rundown}
When \Rundown is executed, all packets succeed within $O(n\sqrt{\CollCost}\ln(n))$ slots.
\end{lemma}
\begin{proof}
By Lemmas \ref{l:okay-GR}, \ref{l:limbo-low}, and \ref{l:limbo-high}, \Rundown will execute with a window size $w \in [10n\sqrt{\CollCost}, $ $10^5 n\sqrt{\CollCost}]$. By Corollary~\ref{cor:rundown-almost-all-packets-succeed}, \Rundown reduces the number of remaining packets to  $O(\log n)$ and uses $r=\lg(n) - \lg\lg(n) + O(1)$ windows in the run to do so. Thus, the number of slots \Rundown executes over is at most:
\begin{align*}
    \sum_{j=0}^{r} \frac{10^5 n \sqrt{\CollCost}}{2^j} & = O\left(n\sqrt{\CollCost} \right)
\end{align*}
\noindent by the sum of a geometric series.    

Any remaining packets, of which there are $O(\log n)$, execute over a window of size $w_0 \in [10n\sqrt{\CollCost}, 10^5 n\sqrt{\CollCost}]$. By the same analysis used to reach Equation \ref{eqn:prob-fail} in the proof of Lemma \ref{l:single-rundown}, each packet succeeds with probability at least $0.7$. 

Repeating this  $c\ln(w_0) \leq c \ln(n)$ times guarantees that all remaining packets succeed with an error bound of at most:
\begin{align*}
    \left(\frac{7}{10}\right)^{c\ln(n)} &\leq n^{c\ln(7/10)}\\
    & \leq \frac{1}{n^{d'}}
\end{align*}
\noindent  so long as $c$ is  sufficiently large. Finally, the number of slots incurred by this component of \Rundown is $O(\log^2(n))$. \qed
\end{proof}

\noindent The next lemma establishes an upper bound on the expected collision cost for \Rundown.

\begin{lemma}\label{lem:expected-cost-collision-rundown}
The expected collision cost for executing \Rundown is $O(n\sqrt{\CollCost}\ln(n))$.
\end{lemma}
\begin{proof}
 By Lemmas \ref{l:okay-GR}, \ref{l:limbo-low}, and \ref{l:limbo-high},  \Rundown first executes with $m=n$ and $w_0 \in [10n\sqrt{\CollCost},$ $10^5 n\sqrt{\CollCost}]$.  By Lemma~\ref{l:single-rundown}, in each window of a run, the number of packets is reduced by at least a factor of $2$, and then the window size halves. That is, at the start of the $i$-th run, the number of packets is $m_i\leq n/2^i$ and the window size is $w_i = w_0/2^i$, for $i\geq 0$.

%; note that we may use this lemma since $w_i \geq  \lg(w_0)\sqrt{\CollCost}/2$ by Line  \ref{alg:while} of \mainAlgorithm, and so $p<1/m$ is 

We employ Lemma~\ref{lem:upper-prob-coll}, where $p_i = 2/w_i$. We have that the expected number of collisions over a window of size $w_i$ in the run is at most:
\begin{align*}
\left(\frac{2m_i^2 p_i^2}{1-m_i p_i}\right)  w_i    & =   \left(\frac{2m_i^2p_i^2}{1-m_ip_i}\right) \left(\frac{2}{p_i}\right)\\
& = \frac{4m_i^2 p_i}{1-m_ip_i}\\
& = O(m_i^2 p_i) 
\end{align*}

By Corollary~\ref{cor:rundown-almost-all-packets-succeed}, there are $r=\lg(n) - \lg\lg(n) + O(1)$ windows in the run. Summing up over these windows, by linearity of expectation, the total expected number of collisions is:
\begin{align*}
    O\left(\sum_{i=0}^{r} m_i^2 p_i  \right) &=  O\left(\sum_{i=0}^{r} \left(\frac{n}{2^i}\right)^2 \left(\frac{2}{w_i}\right)  \right)\\
& =O\left(\sum_{i=0}^{r} \left(\frac{n}{2^i}\right)^2 \left(\frac{2^{i+1}}{n\sqrt{\CollCost}}\right) \right)\\
& =O\left(\sum_{i=0}^{r} \left(\frac{n}{2^i \sqrt{\CollCost}}\right)  \right)\\
& = O\left(  \frac{n}{\sqrt{\CollCost}}  \sum_{i=0}^{r} \frac{1}{2^i}\right)\\
& = O\left(  \frac{n}{\sqrt{\CollCost}} \right)
\end{align*}
\noindent where the last line follows from the sum of a geometric series. Thus, multiplying this by $\CollCost$, we obtain that the expected collision cost is $O(n\sqrt{\CollCost})$.

Finally, in lines \ref{alg:repeat-runs} to \ref{alg:repeat-rsingle-window-rundown}, the remaining $O(\log n)$ packets execute over a window of size $w_0$, doing so $O(\ln (w_0) = O(\ln(n))$ times before succeeding, as established by Lemma~\ref{lem:correctness-latency-rundown}. Certainly, each of the $O(\ln (w_0))$ executions has expected collision cost no more than that incurred by lines \ref{alg:while} to \ref{alg:run-end} (where we have $n$ packets executing), and thus the expected collision cost is $O(n\sqrt{\CollCost}\ln(w_0))=O(n\sqrt{\CollCost}\ln(n))$. \qed
\end{proof}

\subsection{Analysis of \mainAlgorithm}\label{sec:main-alg}

\noindent We wrap up our analysis of \malg by calling on our main lemmas from the previous sections. \smallskip

\noindent{\bf Theorem~\ref{thm:main-theorem}}
{\it    W.h.p. \mainAlgorithm guarantees that all packets succeed, the makespan is $O(n\sqrt{\CollCost}\log(n) )$, and the expected collision cost is $O(n\sqrt{\CollCost} \log^2(n))$.}\smallskip
\begin{proof}
   By Lemma \ref{lem:correctness-cost-collectsample-diagnosis}, over the course of executing \Sample and \Diagnosis,  \whp \malg executes over  $O(\sqrt{\CollCost}\log^2(n) )$ slots and the expected collision cost is $O(n\sqrt{\CollCost} \log^2(n))$. By Lemmas \ref{lem:correctness-latency-rundown} and \ref{lem:expected-cost-collision-rundown},\whp when \Rundown is executed in \malg, all packets succeed within $O(n\sqrt{\CollCost}\ln(n))$  slots and the expected collision cost is $O(n\sqrt{\CollCost}\ln(n))$. Therefore, \malg guarantees that all packets succeed, the makespan is $O(n\sqrt{\CollCost}\log(n) )$, and the expected collision cost is $O(n\sqrt{\CollCost} \log^2(n))$. \qed
\end{proof}

\section{Lower Bound}\label{app:lowerbound}
     
We consider only the set of slots, {\boldmath{$\twoActiveSlots$}}, in the execution of any algorithm where at least two active packets remain, since we cannot have a collision with a single active packet. While we do not always make this explicit, but going forward, any slot $t$ is assumed to implicitly belong to $\twoActiveSlots$.

Let {\boldmath{$p_i(t)$}} denote the probability that packet $i$ sends in slot $t$.  Note that, if a packet has terminated, its sending probability can be viewed as $0$. For any fixed slot $t$, the \defn{contention} in slot $t$ is  {\bf \texttt{Con}}{\boldmath{$(t)$}} $= \sum_{i=1}^n p_i(t)$; that is, the sum of the sending probabilities in that slot.

We begin by arguing that any algorithm with contention exceeding $2$ in any slot must incur an expected collision cost of $\Omega(\CollCost)$. We do this by deriving an upper bound on the probability of (i) a success and (ii) an empty slot, as a function of contention. In turn, this provides a lower bound on the probability of a collision that is useful when contention exceeds $2$, allowing us to show $\Omega(\CollCost)$ expected cost in Lemma \ref{lem:omega-collision}.

% \umesh{We have proved this in lemma 4, do we need to refer that?}\maxwell{we are foreshadowing/telling the reader what is coming, although it takes us a couple lemmas to get there. Modifying this to make it more clear}.

\begin{lemma}\label{lem:prob-success}
Fix any slot $t$. The probability of a success is at most $\frac{e\cdot \con}{e^{\con}}$.
\end{lemma}
\begin{proof}
Fix a slot $t$.  For any sending probabilities $p_i(t)$, the probability of a success in slot $t$ is:
\begin{align*} 
  & ~~~~\sum_{\mbox{\tiny all $j$}} \left(p_j(t) \prod_{\mbox{\tiny all  $i\not=j$}}(1-p_i(t))\right)\\
   & \leq \sum_{\mbox{\tiny all  $j$}} \left(p_j(t) e^{-\sum_{\mbox{\tiny all $i\not=j$ }} p_i(t)}\right) \mbox{~~~by Fact \ref{fact:taylor}(a)}\\
    & \leq \sum_{\mbox{\tiny all  $j$}} \left(p_j(t) e^{-(\con-1)}\right) \mbox{~~~since $p_j(t)\leq 1$}\\
  & \leq  \frac{\con}{e^{\con-1}}\\
 & \leq  \frac{e\con}{e^{\con}}
  \end{align*}
  \noindent which completes the proof.\qed
\end{proof}

\begin{lemma}\label{lem:high-contention}
Fix any slot $t$. If $\con> 2$, then the probability of a collision is at least $1/10$.
\end{lemma}
\begin{proof}
By Lemma \ref{lem:prob-success}, the probability of a success in slot $t$ is at most $\frac{e\con}{e^{\con}}$. The probability that slot $t$ is empty is at most $\prod_{\mbox{\tiny all j}} (1-p_j(t)) \leq e^{-\con}$ by Fact \ref{fact:taylor}(a). Therefore, the probability of a collision is at least:
\begin{align*}
   1 -  \frac{e\,\con}{e^{\con}} - e^{-\con} & \geq 1 -  \frac{2e}{e^{2}} - e^{-2} \mbox{~~~since $\con> 2$}\\
   & \geq 1/10.
\end{align*}
\noindent as claimed.\qed
\end{proof}

\begin{lemma}\label{lem:omega-collision}
    Any algorithm that has $\con> 2$ has an expected collision cost that is $\Omega(\CollCost)$.
\end{lemma}
\begin{proof}
   Let $t$ be any slot where $\con> 2$ in the execution of any algorithm. By Lemma  \ref{lem:high-contention}, the probability of a collision is at least $1/10$. Therefore, the expected cost from collisions is at least $\CollCost/10$.\qed %Since $n$ packets require $n$ slots to succeed, the total expected cost is at least $\Omega(n+\CollCost)$.
\end{proof}

Next, we need to consider algorithms with ``low'' contention in every slot; that is,  $\con \leq 2$ for all $t$. As a stepping stone, in Lemmas \ref{lem:lower-prob-coll} we derive a lower bound on the probability of a collision that is a function of contention, and which is easy to work with when contention is between $0$ and $2$.\medskip

\begin{lemma}\label{lem:lower-prob-coll} 
Fix any slot $t$ and let $\con\leq 2$. The probability of a collision in $t$ is at least $(\frac{1}{110})\left(\con^2 - \sum_i p_i(t)^2 \right)$.   
\end{lemma} 
\begin{proof} 
For ease of presentation, in this proof we write $p_i$ instead of $p_i(t)$. We restrict our analysis to collisions involving only two packets; note that considering collisions involving three or more packets can only increase the probability of a collision. 

For the following analysis, we restrict all sending probabilities $p_i$ are at most $1/2$; however, we will remove this restriction momentarily. The probability of collision in slot $t$ is at least:
\begin{align*}
& \frac{1}{2}\bigg( p_1p_2 \prod_{i=1}^n (1-p_i) + p_1p_3\prod_{i=1}^n (1-p_i) + ... + p_1p_n \prod_{i=1}^n (1-p_i)+  \\
& \mbox{~~~\hspace{5pt}} p_2p_1 \prod_{i=1}^n (1-p_i) + p_2p_3\prod_{i=1}^n (1-p_i) + ... + p_2p_n \prod_{i=1}^n (1-p_i)+\\
& \vdots\\
& \mbox{~~~\hspace{5pt}}   p_np_1 \prod_{i=1}^n (1-p_i) + p_np_2\prod_{i=1}^n (1-p_i) + ... + p_n p_{n-1} \prod_{i=1}^n (1-p_i)  \bigg)
\end{align*}
\noindent{}where the leading factor of $1/2$ comes from double counting terms (e.g., $p_1p_2$ and $p_2p_1$), which allows us to simply below. Since each sending probability is at most $1/2$, we can use Fact \ref{fact:taylor}(c) to simplify in the following way.
\begin{align*}
& \geq \left(\frac{1}{2}\right) {e^{-2 \con}} \bigg(p_1 \sum_{j\neq1} p_j + p_2 \sum_{j\neq2} p_j + \cdots + p_n \sum_{j\neq n} p_j \bigg)\\
& = \left(\frac{1}{2}\right) {e^{-2 \con}} \hspace{-3pt}\left(p_1 \hspace{-3pt} \left(\hspace{-2pt}\left(\sum_{\mbox{\tiny all $j$}} p_j\right) - p_1\hspace{-3pt}\right) \hspace{-2pt}+\hspace{-1pt} p_2 \hspace{-2pt}\left(\hspace{-2pt}\left(\sum_{\mbox{\tiny all $j$}} p_j\right) - p_2\hspace{-2pt}\right) \hspace{-2pt}+ \cdots + p_n\hspace{-2pt} \left(\hspace{-2pt}\left(\sum_{\mbox{\tiny all $j$}} p_j\right) - p_n\hspace{-2pt}\right)\hspace{-2pt}\right)\\
& = \left(\frac{1}{2}\right) e^{-2 \con} \left( \left(p_1 \con - p_1^2\right) + \left(p_2 \con - p_2^2\right) + \cdots + \left(p_n \con - p_n^2\right)  \right)\\
& = \left(\frac{1}{2}\right) e^{-2 \con} \left( \con (p_1 + p_2 + \cdots + p_n) - \sum p_j^2\right)\\
& \geq \frac{\con^2 - \sum_i p_i^2}{2e^{2 \con}}\\
& \geq  \frac{\con^2 - \sum_i p_i^2}{110} \mbox{~~~since $\con \leq 2$.}
\end{align*}
\noindent Finally, note that if we allow the sending probabilities to be larger than $1/2$, this can only increase the probability of a collision; therefore, our lower bound holds even when the sending probabilities are not restricted.\qed
\end{proof}

Before making our next argument, consider a simple example with two packets, whose sending probabilities are $p_1(t)=1$ and $p_2(t)=1/100$ in slot $t$. Observe that $\con^2 - (p_1^2 + p_2^2) = 1/50 \ll \con^2$. However, here $\Delta(t)=1/100$, and it is true that $\con^2 - (p_1^2 + p_2^2) \geq \Delta(t) \con^2$. Lemma \ref{lem:con-minus-squares} below generalizes this inequality.\medskip

\begin{lemma} \label{lem:con-minus-squares}
For any fixed slot $t$,  $\con^2 - \sum_i p_i(t)^2 \geq  \Delta(t)\, \con^2/2$.
\end{lemma}
\begin{proof}
For ease of presentation, in this proof we write $p_i$ instead of $p_i(t)$. Fix a slot $t$. Without loss of generality let the summands of $\con$ be labeled in non-monotonic decreasing order: $p_1 \geq p_2 \geq ... \geq p_n$. For part one of our argument, we note that:
\begin{align}
p_1 p_2  \geq p_1^2\Delta(t)  \label{eq:probs}
\end{align}
\noindent since $\Delta(t)  = p_2/p_1$.  

Part two of our argument proceeds as follows. For $i\geq 2$, for any $p_i^2$ term removed from $\con^2$, we focus on a `leftover'' term in $\con^2 - \sum_i p_i^2$ of the form $p_{i-1}p_i$. For example, for $p_2^2$ we have a leftover term $p_1 p_2$, and we note that $p_1 p_2 \geq p_2^2$, since $p_1 \geq p_2$. Generalizing, the leftover term  $p_{i-1}p_i$ satisfies $p_{i-1}p_i\geq p_i^2$, since $p_{i-1}\geq p_i$. 

We tie things together by first noting:
\begin{align*}
 \con^2 - \sum_i p_i^2 &\geq p_1 p_2 +  p_1p_2 + p_2p_3 + p_3p_4 +...+p_{n-1}p_n
 \end{align*}
\noindent where we highlight that we need both $p_1p_2$ terms; one to offset the subtraction of $p_1^2$ from $\con$, as discussed in part one of our argument, and the other to offset the subtraction of $p_2^2$ from $\con$, as discussed in part two of our argument. For the subtraction of $p_i^2$ for $i\geq 3$, we need only a single corresponding term $p_{i-1}p_i$ to offset this. Thus, continuing from above, we have:
 \begin{align}
 &\mbox{~~~~~~} p_1 p_2 +  p_1p_2 + p_2p_3 + p_3p_4 +...+p_{n-1}p_n \nonumber\\
 &\geq p_1^2\Delta(t)  +  p_1p_2 + p_2p_3 + p_3p_4 +...+p_{n-1}p_n \mbox{~~~by Equation \ref{eq:probs}}\nonumber \\
 & \geq p_1^2\Delta(t)  + p_2^2 + p_3^2 + ... + p_n^2  \mbox{~~~since $p_{i-1}p_i \geq p_i^2$ for $i\geq 2$} \nonumber\\
  & \geq \Delta(t) \sum_i p_i^2  \mbox{~~~since $\Delta(t)\leq 1$}\label{eq:delta-con}
\end{align}
If $\sum_i p_i^2 \geq \con^2/2$, then  by the above:
\begin{align*}
    \con^2 - \sum_i p_i^2 & \geq \Delta(t) \sum_i p_i^2 \mbox{~~~by Equation \ref{eq:delta-con}}\\ 
   & \geq  \Delta(t)\, \con^2/2.
\end{align*}
\noindent Else, $\sum_i p_i^2 < \con^2/2$ and  $\con^2 - \sum_i p_i^2 \geq \con^2/2 \geq \Delta(t) \con^2/2$, given $\Delta(t)\leq 1$. Either way, the claim follows.\qed
\end{proof}
 
For any algorithm $\mathcal{A}$ that executes over a set of slots $\mathcal{I}$, we define \defn{total contention} of $\mathcal{A}$ to be $\sum_{s\in \mathcal{I}} \con$; that is, the sum of contention over all slots of $\mathcal{A}$'s execution. Note, for this definition, we are counting those slots during which a single packet is active (for simplicity, we do not discuss this intermediate result in the overview of Section \ref{sec:lower-bound}, so we were able to restrict our attention to only the set $\twoActiveSlots$).\medskip

\begin{lemma}\label{lem:sum-contention}
Any algorithm that \whp has $n$ packets succeed has the property that the total contention over all slots in the execution is at least $n/16$.
\end{lemma}
\begin{proof}
Assume the existence of an algorithm $\mathcal{A}$ that guarantees with probability at least $1-1/n$ that all $n$ packets succeed, but has the property that the sum of contention over all slots of the execution is less than $n/16$.

For any packet $u$, define the {\it lifetime contribution} of $u$ to be the sum of $u$'s sending probabilities over all slots of the execution of $\mathcal{A}$; denote this by {\boldmath{$\LC(u)$}}. Call packet $u$ ``heavy'' if $\LC(u) \geq 1/4$; otherwise, $u$ is ``light''. 

There must be less than $n/4$ heavy packets. Otherwise, the sum of the contention for all heavy packets is at least $(n/4)(1/4) \geq n/16$, which contradicts the existence of $\mathcal{A}$. Therefore, at least $(3/4)n$ packets are light.

Consider any light packet $u$. Let $p(t)$ denote the probability of that $u$ sends in slot $t$ specified arbitrarily by $\mathcal{A}$. The probability that $u$ succeeds in slot $t$ is at most $p(t)$; therefore, the probability that $u$ fails in slot $t$ is at least $1-p(t)$. Over the execution of $\mathcal{A}$, the probability that $u$ does not succeed is at least:
\begin{align*}
\prod_t (1-p(t)) &\geq e^{-2\sum_t p(t)}\mbox{~~~by Fact \ref{fact:taylor}(c)}\\
& \geq e^{-2 \cdot \LC(u)}.
\end{align*}
\noindent Since $\LC(u) < 1/4$, then the probability of failure is at least:
\begin{align*}
e^{-2 \LC(u) } &\geq 1 - 2\,\LC(u) \mbox{~~~by Fact \ref{fact:taylor}(a)} \\
& \geq 1 -  2(1/4)\\
& \geq 1/2.
\end{align*}
Thus, there is a packet that with probability at least $1/2$ does not succeed over the execution of $\mathcal{A}$, which is a contradiction. \qed
\end{proof}

\begin{lemma}\label{lem:last-packet-contention}
Consider any algorithm that guarantees \whp that all packets succeed. Consider the portion of the execution during which only one active packet remains. W.h.p., the sum of the contention over this portion of the execution is $O(\log n)$.
\end{lemma}
\begin{proof}
Let $u$ be the last packet to succeed. Let $\mathcal{I}$ be the set of slots during which $u$ is the only remaining active packet. Let $p(t)$ denote the probability of that $u$ sends in slot $t \in \mathcal{I}$ specified arbitrarily by the algorithm. Since it is the only active packet, the probability that $u$ succeeds in slot $t$ is $p(t)$; therefore, the probability that $u$ fails in slot $t$ is  $1-p(t)$. Over the execution of $\mathcal{A}$, the probability that $u$ does not succeed is at most:
\begin{align*}
\prod_{t\in\mathcal{I}} (1-p(t)) & \leq  e^{-\sum_{t\in\mathcal{I}} p(t)} \mbox{~~~by Fact \ref{fact:taylor}(a)}
%& \leq e^{-d\ln n} 
\end{align*}
\noindent Note that for $\sum_{t\in \mathcal{I}} p(t) = O(\log n)$, this probability of failure is polynomially small in $n$. \qed
\end{proof}

\begin{lemma}\label{lem:sum-contention-restricted}
Any algorithm that guarantees \whp that $n$ packets succeed must \whp ~ have  $\sum_{t\in\twoActiveSlots} \con$ $= \Omega(n)$.
\end{lemma} 
\begin{proof}
By definition, $\sum_{t\in\twoActiveSlots} \con$ is at least the total contention in Lemma \ref{lem:sum-contention} minus the sum of the contention in Lemma \ref{lem:last-packet-contention}, which is at least: $(n/16) - O(\log n) = \Omega(n)$. \qed
\end{proof}

\begin{lemma}\label{lem:lower-bound}
Consider any algorithm $\mathcal{A}$ whose contention in any slot is at most $2$ and guarantees \whp a makespan of $\tilde{O}(n\sqrt{\CollCost})$.  W.h.p. the expected collision cost for $\mathcal{A}$ is $\tilde{\Omega}(\deltaMin n\sqrt{\CollCost})$.
\end{lemma}
\begin{proof}
Let $X$ be a random variable that is $|\twoActiveSlots|$ under the execution of $\mathcal{A}$. Note that, since \whp the makespan is $O(n\sqrt{\CollCost})$, it must be the case that \whp $X = O(n\sqrt{\CollCost})$. 

By Lemma \ref{lem:sum-contention-restricted}, we have:
\begin{align}
    \sum_{t=1}^{X}  \con \geq c n. \label{eq:contention-sum}
\end{align}
\noindent for some constant $c>0$. Let $Y_{t}=1$ if slot $t \in \twoActiveSlots$  has a collision; otherwise, $Y_{t}=0$. By Lemmas \ref{lem:lower-prob-coll} and \ref{lem:con-minus-squares}:
\begin{align*}
 Pr(Y_t=1) & \geq   \Delta(t) \cdot \con^2/220 \nonumber\\
 & \geq \deltaMin \cdot \con^2/220 
\end{align*}
where the second line  follows from the definition of $\deltaMin$. The expected collision cost is:
\begin{align}
     \sum_{t=1}^X P(Y_t=1)\cdot \CollCost & \geq \sum_{t=1}^X \frac{\deltaMin \con^2 \cdot \CollCost} {220} \nonumber\\
     & = \frac{\deltaMin \cdot \CollCost}{220} \sum_{t=1}^X \con^2. \label{eq:collision-sum}
\end{align}
By Jensen's inequality for convex functions, we have:
\begin{align}
    \frac{\sum_{t=1}^X \con^2}{X} \geq \left(\frac{\sum_{t=1}^X \con}{X}\right)^2 \label{eq:jensen}
\end{align}  
Finally, the expected cost is at least:
\begin{align*}
    \frac{\deltaMin \cdot \CollCost}{220} \sum_{t=1}^X \con^2 & \geq   \frac{\deltaMin \cdot \CollCost}{220} \frac{\left(\sum_{t=1}^X \con\right)^2}{X} \mbox{~~~by Equations \ref{eq:collision-sum} and \ref{eq:jensen}}\\
    & \geq \frac{\deltaMin \cdot \CollCost}{220} \left(\frac{c^2 n^2}{X}\right) \mbox{~~~by Equation \ref{eq:contention-sum}}\\
    & = \tilde{\Omega}\left( \deltaMin\, n\sqrt{\CollCost}  \right)
\end{align*}
\noindent where the second line follows by Equation \ref{eq:contention-sum},  which was defined with regard to $\twoActiveSlots$, and so can be compared to Equation  \ref{eq:collision-sum}. The last line follows since \whp $X= \tilde{O}(n\sqrt{\CollCost})$.\qed
\end{proof}

\noindent We now state our lower bound, which follows directly from Lemmas \ref{lem:omega-collision} and \ref{lem:lower-bound}.

\begin{theorem}\label{thm:general-lower-bound-main}
Consider any algorithm that \whp guarantees $\tilde{O}(n\sqrt{\CollCost})$ makespan. Then, \whp, the expected collision cost for $\mathcal{A}$ is $\tilde{\Omega}(\min\{\CollCost, \deltaMin n\sqrt{\CollCost}\})$.   
\end{theorem}

\noindent Addressing Theorem \ref{thm:main-lower-bound}, note that for fair algorithms, $\deltaMin=1$. Recall that our analysis ignores any slot $t$ where all packets have sending probability $0$ (i.e.,$\con=0$), giving such slots to the algorithm for free. Given that $\deltaMin=1$,  $\tilde{\Omega}(\min\{\CollCost, \deltaMin n\sqrt{\CollCost}\}) = \tilde{\Omega}(n\sqrt{\CollCost})$ when $\CollCost = c' n^2$ for a sufficiently large positive constant $c'$.

\end{document}